\documentclass[final,3p]{elsarticle}
\usepackage{lineno,hyperref}
\usepackage{amsmath}
\usepackage{amssymb}
\usepackage[linesnumbered,ruled,vlined]{algorithm2e}
\usepackage{caption}
\usepackage{mathtools}
\usepackage{changes}
\usepackage{multirow}
\usepackage{verbatim}
\usepackage{mathrsfs}
\usepackage{graphicx}
\usepackage{subcaption}
\usepackage{amsmath,amsthm,bm,mathrsfs}
\theoremstyle{plain}
\newtheorem{theorem}{Theorem}[section]

\newtheorem{proposition}[theorem]{Proposition}

\theoremstyle{definition}

\theoremstyle{remark}

\modulolinenumbers[5]

\journal{ArXiv.org}

\begin{document}

\begin{frontmatter}

\title{A New Low-Rank Cholesky-Factor ADI Algorithm Allowing Shifts Anywhere in the Complex Plane with Applications to Data-Driven Model Reduction}

\author[uz]{Umair~Zulfiqar\corref{mycorrespondingauthor}}
\cortext[mycorrespondingauthor]{Corresponding author}
\ead{umair@yangtzeu.edu.cn}
\address[uz]{School of Electronic Information and Electrical Engineering, Yangtze University, Jingzhou, Hubei, 434023, China}
\begin{abstract}
The low-rank Cholesky factor alternating direction implicit (LRCF-ADI) iteration method \cite{benner2013efficient,benner2013reformulated} is an effective and efficient approach for computing low-rank solutions to large-scale Lyapunov equations in the form \(P\approx ZZ^\top\). This form is useful for balanced truncation, as the square-root algorithm requires computing the controllability and observability Gramians in this form. The standard LRCF-ADI method requires all ADI shifts to have negative real parts, which can be restrictive for applications like frequency-limited and data-driven balanced truncation, where purely imaginary ADI shifts are a more suitable choice. This paper proposes a new LRCF-ADI method where the ADI shifts can be located anywhere in the complex plane, including on the imaginary axis. The proposed generalized LRCF-ADI algorithm reduces to the standard LRCF-ADI algorithm as a special case. The new method is also extended to solve frequency-limited Lyapunov equations, time-limited Lyapunov equations, and Riccati equations. Approximations of matrix logarithm and matrix exponential products using the proposed method are also discussed.

LRCF-ADI-based reduced models for balanced truncation can be constructed non-intrusively from transfer function samples at the mirror images of the ADI shifts, without accessing the state-space realization of the original model. Since the standard LRCF-ADI method requires all ADI shifts to have negative real parts, its non-intrusive implementation requires samples in the right half of the complex plane, which cannot be measured in an experimental setting. However, the ADI shifts in the proposed method can lie on the imaginary axis. Exploiting this property, we also propose a data-driven low-rank balanced truncation algorithm that requires only transfer function samples on the imaginary axis, which can be measured experimentally. The effectiveness of the proposed algorithms is demonstrated through benchmark numerical examples. 
\end{abstract}


\end{frontmatter}

\section{Introduction}
Consider the $n$th-order linear time-invariant (LTI) dynamical system $G(s)$ with state-space realization
\begin{align}
G(s)=C(sE-A)^{-1}B+D,\nonumber
\end{align}
where $E\in\mathbb{R}^{n\times n}$, $A\in\mathbb{R}^{n\times n}$, $B\in\mathbb{R}^{n\times m}$, $C\in\mathbb{R}^{p\times n}$, and $D\in\mathbb{R}^{p\times m}$. Throughout this paper, we assume that $E$ is invertible and that all eigenvalues of $E^{-1}A$ lie in the left half of the complex plane, i.e., $E^{-1}A$ is Hurwitz.

The controllability Gramian $P$ and observability Gramian $Q$ associated with the realization $(A,B,C,D,E)$ satisfy the following Lyapunov equations:
\begin{align}
APE^\top+EPA^\top +BB^\top &=0,\label{lyap_p}\\
A^\top QE+E^\top Q A +C^\top C &=0.
\end{align}
If $G(s)$ is symmetric, i.e., $G(s)=G^\top(-s)$, then the cross Gramian associated with $(A,B,C,D,E)$ satisfies the following Sylvester equation:
\begin{align}
AXE+EXA+BC=0,\label{cross_sylv}
\end{align}
where $X^2=PQ$.
\subsection{The low-rank Cholesky factor alternating direction implicit (LRCF-ADI) method for Lyapunov equations \cite{benner2013efficient,benner2013reformulated}}
Let $\{\alpha_i\}_{i=1}^{k}\subset\mathbb{C}_{-}$ be the ADI shifts used in the LRCF-ADI method \cite{benner2013efficient,benner2013reformulated} to approximate the Lyapunov equation \eqref{lyap_p}. Define $v_i^{\mathrm{adi}}$ by
\[
v_i^{\mathrm{adi}} = (A + \alpha_i E)^{-1}B_{\perp,i-1}^{\mathrm{adi}},
\]
where 
\[
B_{\perp,i}^{\mathrm{adi}} = B_{\perp,i-1}^{\mathrm{adi}} - 2\operatorname{Re}(\alpha_i) \, E v_i^{\mathrm{adi}},
\]
with $B_{\perp,0}^{\mathrm{adi}}=B$.

Next, define $Z_p^{(k)}$ by
\[
Z_p^{(k)}=\begin{bmatrix}\sqrt{-2\operatorname{Re}(\alpha_1)}v_1^{\mathrm{adi}}&\cdots&\sqrt{-2\operatorname{Re}(\alpha_k)}v_k^{\mathrm{adi}}\end{bmatrix}.
\]
Then, the low-rank approximation $P\approx\tilde{P}=Z_p^{(k)}(Z_p^{(k)})^*$ obtained from the LRCF-ADI method \cite{benner2013reformulated} satisfies the residual relation
\[
A\tilde{P}E^\top+E\tilde{P}A^\top +BB^\top = B_{\perp,k}^{\mathrm{adi}}(B_{\perp,k}^{\mathrm{adi}})^\top.
\]
\subsection{Factored ADI (fADI) method for Sylvester equations \cite{tu2009adi,benner2014computing}}
Let $\{\alpha_i\}_{i=1}^{k}\subset\mathbb{C}$ and $\{\beta_i\}_{i=1}^{k}\subset\mathbb{C}$ be ADI shifts satisfying $\alpha_i\neq-\beta_i$, used to approximate the Sylvester equation \eqref{cross_sylv} with the fADI method \cite{tu2009adi,benner2014computing}. Define $v_i^{\mathrm{fadi}}$ and $w_i^{\mathrm{fadi}}$ by
\begin{align}
v_i^{\mathrm{fadi}}&=(A + \alpha_i E)^{-1}B_{\perp,i-1}^{\mathrm{fadi}},\nonumber\\
w_i^{\mathrm{fadi}} &=(A^\top + \overline{\beta_i} E^\top)^{-1}(C_{\perp,i-1}^{\mathrm{fadi}})^*,\nonumber
\end{align}
where
\begin{align}
B_{\perp,i}^{\mathrm{fadi}} &= B_{\perp,i-1}^{\mathrm{fadi}} - (\alpha_i + \beta_i) E v_i^{\mathrm{fadi}},\nonumber\\
C_{\perp,i}^{\mathrm{fadi}} &= C_{\perp,i-1}^{\mathrm{fadi}} - (\alpha_i + \beta_i) (w_i^{\mathrm{fadi}})^* E,\nonumber
\end{align}
with $B_{\perp,0}^{\mathrm{fadi}}=B$ and $C_{\perp,0}^{\mathrm{fadi}}=C$.

Next, define $V_{\mathrm{fadi}}^{(k)}$, $W_{\mathrm{fadi}}^{(k)}$, and $J_{\mathrm{fadi}}^{(k)}$ by
\begin{align}
V_{\mathrm{fadi}}^{(k)}=\begin{bmatrix}v_1^{\mathrm{fadi}}&\cdots&v_k^{\mathrm{fadi}}\end{bmatrix},\quad
W_{\mathrm{fadi}}^{(k)}=\begin{bmatrix}w_1^{\mathrm{fadi}}&\cdots&w_k^{\mathrm{fadi}}\end{bmatrix},\quad
J_{\mathrm{fadi}}^{(k)}=\begin{bmatrix}j_{\mathrm{fadi}}^{(1)}&\cdots&0\\\vdots&\ddots&\vdots\\0&\cdots&j_{\mathrm{fadi}}^{(k)}\end{bmatrix},\nonumber
\end{align}
where $j_{\mathrm{fadi}}^{(i)}=-(\alpha_i+\beta_i)$.

Then, the low-rank approximation $X\approx\tilde{X}=V_{\mathrm{fadi}}^{(k)}(J_{\mathrm{fadi}}^{(k)}\otimes I_m)(W_{\mathrm{fadi}}^{(k)})^*$ obtained from the fADI method satisfies the residual relation
\[
A\tilde{X}E+E\tilde{X}A+BC=B_{\perp,k}^{\mathrm{fadi}}C_{\perp,k}^{\mathrm{fadi}}.
\]
\section{A Generalized LRCF-ADI Algorithm}
In this section, the condition that the shifts $\alpha_i$ must have negative real parts is dropped, while $P$ is still approximated in the form $P\approx ZZ^*$, leading to a Generalized LRCF-ADI (G-LRCF-ADI) algorithm. To this end, the fADI method \cite{tu2009adi,benner2014computing} is used to approximate \eqref{lyap_p}, requiring only the condition $\alpha_i\neq-\beta_i$. The equivalent realified version of the G-LRCF-ADI algorithm, which approximates $P$ as $P\approx ZZ^\top$, is also presented. It is further shown that G-LRCF-ADI can be used to approximate several other matrix equations and products.

Define $\hat{A}^{(k)}$, $\hat{B}^{(k)}$, and $\hat{C}^{(k)}$ by
\begin{align}
\hat{A}^{(k)}=\begin{bmatrix}\hat{A}^{(k-1)}&0\\A_{21}^{(k)}&\beta_kI_m\end{bmatrix},\quad \hat{B}^{(k)}=\begin{bmatrix}\hat{B}^{(k-1)}\\-j_{\mathrm{fadi}}^{(k)}I_m\end{bmatrix},\quad
\hat{C}^{(k)}=CV_{\mathrm{fadi}}^{(k)},\label{interpolant1}
\end{align}
where $A_{21}^{(k)}=(-j_{\mathrm{fadi}}^{(k)}\mathbf{1}_{k-1}^\top)\otimes I_m$ and $\mathbf{1}_{k}$ is a $k\times 1$ column vector of ones \cite{zulfiqar2026unified}. A MATLAB function named \texttt{compute\_A\_B} is provided in \cite{mycodes} to construct $(\hat{A}^{(k)},\hat{B}^{(k)})$ from the shifts $\alpha_i$ and $\beta_i$.

Next, define $\hat{G}^{(k)}(s)$ by $\hat{G}^{(k)}(s)=\hat{C}^{(k)}(sI-\hat{A}^{(k)})^{-1}\hat{B}^{(k)}+D$. Assuming that $V_{\mathrm{fadi}}^{(k)}$ is full column rank, the following interpolation conditions hold:
\[
G(-\alpha_i)=\hat{G}^{(k)}(-\alpha_i),\quad i=1,\ldots,k.
\]
The reduced-order model (ROM) $\hat{G}^{(k)}(s)$ has poles at $\beta_i$ \cite{zulfiqar2026unified}.

Thus, when the shifts $\beta_i$ have negative real parts, i.e., $\{\beta_i\}_{i=1}^{k}\subset\mathbb{C}_{-}$, the projected Lyapunov equation
\begin{align}
\hat{A}^{(k)}\hat{P}^{(k)}+\hat{P}^{(k)}(\hat{A}^{(k)})^*+\hat{B}^{(k)}(\hat{B}^{(k)})^*=0\label{proj_lyap}
\end{align}
has a unique solution.

We now show that $\hat{P}^{(k)}$ can be factorized as $\hat{P}^{(k)}=Z^{(k)}(Z^{(k)})^*$ and that an analytical formula for $Z^{(k)}$ exists, depending solely on the ADI shifts $\alpha_i$ and $\beta_i$ when $\{\beta_i\}_{i=1}^{k}\subset\mathbb{C}_{-}$.

Define $S_{\mathrm{fadi}}^{(k)}$ and $L_{\mathrm{fadi}}^{(k)}$ by
\begin{align}
S_{\mathrm{fadi}}^{(k)}=\begin{bmatrix}S_{\mathrm{fadi}}^{(k-1)}&-(J_{\mathrm{fadi}}^{(k-1)})^*(L_{\mathrm{fadi}}^{(k-1)})^\top\\0&s_{\mathrm{fadi}}^{(k)}\end{bmatrix}\quad\text{and}\quad L_{\mathrm{fadi}}^{(k)}=\begin{bmatrix}L_{\mathrm{fadi}}^{(k-1)}&l_{\mathrm{fadi}}^{(k)}\end{bmatrix},
\end{align}
where $s_{\mathrm{fadi}}^{(k)}=-\overline{\beta}_k$ and $l_{\mathrm{fadi}}^{(k)}=-1$.

The matrix $W_{\mathrm{fadi}}^{(k)}$ (which does not play any role in our context) solves the following Sylvester equation:
\[
A^\top W_{\mathrm{fadi}}^{(k)}-E^\top W_{\mathrm{fadi}}^{(k)} S_{\mathrm{fadi}}^{(k)}+C^\top L_{\mathrm{fadi}}^{(k)}=0;
\]
cf. \cite{zulfiqar2026unified}.
Furthermore, define $S_{\mathrm{adi}}^{(k)}$ and $L_{\mathrm{adi}}^{(k)}$ by
\begin{align}
S_{\mathrm{adi}}^{(k)}=\begin{bmatrix}S_{\mathrm{adi}}^{(k-1)}&(L_{\mathrm{adi}}^{(k-1)})^\top l_{\mathrm{adi}}^{(k)}\\0&s_{\mathrm{adi}}^{(k)}\end{bmatrix}\quad\text{and}\quad L_{\mathrm{adi}}^{(k)}=\begin{bmatrix}L_{\mathrm{adi}}^{(k-1)}&l_{\mathrm{adi}}^{(k)}\end{bmatrix},
\end{align}
where $s_{\mathrm{adi}}^{(k)}=-\overline{\beta_k}$ and $l_{\mathrm{adi}}^{(k)}=-\sqrt{-2\operatorname{Re}(\beta_k)}$.

If the ADI shifts $\alpha_i$ are replaced with $\overline{\beta}_i$ in the LRCF-ADI method, then the matrix $Z_p^{(k)}$ solves the following Sylvester equation:
\[
A Z_p^{(k)}-E Z_p^{(k)} S_{\mathrm{adi}}^{(k)}+B L_{\mathrm{adi}}^{(k)}=0;
\]
cf. \cite{zulfiqar2026unified}.

Let $T^{(k)}$ solve the following Sylvester equation:
\begin{align}
-(S_{\mathrm{adi}}^{(k)})^*T^{(k)}-T^{(k)}S_{\mathrm{fadi}}^{(k)}+(L_{\mathrm{adi}}^{(k)})^*L_{\mathrm{fadi}}^{(k)}=0;\label{trans_sylv}
\end{align}
cf. \cite{zulfiqar2026unified}.

The upper triangular matrix $T^{(k)}$ is given recursively by
\begin{equation}
    T^{(k)}({i,j}) = 
    \begin{cases}
        \displaystyle \frac{-\sqrt{-2\operatorname{Re}(\beta_i)} \left(1 + \sum_{q=1}^{i-1} \big(-\sqrt{-2\operatorname{Re}(\beta_q)}\big) T_{qi}^{(k)}\right)}{\beta_i + \overline{\beta}_i}, & i = j, \\[2ex]
        \displaystyle \frac{-\sqrt{-2\operatorname{Re}(\beta_i)} \left(1 + \sum_{q=1}^{i-1} \big(-\sqrt{-2\operatorname{Re}(\beta_q)}\big) T_{qj}^{(k)}\right) + \sum_{q=i}^{j-1} T_{iq}^{(k)} \overline{-(\alpha_q + \beta_q)}}{\beta_i + \overline{\beta}_j}, & i < j,
    \end{cases}\label{T_analy}
\end{equation}
for $1 \leq i \leq j \leq k$; cf. \cite{zulfiqar2026unified}. A MATLAB function named \texttt{build\_T\_recursive} is provided in \cite{mycodes} to construct $T^{(k)}$ from the shifts $\alpha_i$ and $\beta_i$.

The matrix $T^{(k)}$ can also be written as
\[
T^{(k)}=\begin{bmatrix}T^{(k-1)}&t_1^{(k)}\\0&t_2^{(k)}\end{bmatrix},
\]
where
\begin{align}
t_1^{(k)}(j) &= \frac{-\sqrt{-2\operatorname{Re}(\beta_j)} \left( 1 + \displaystyle\sum_{q=1}^{j-1} \big(-\sqrt{-2\operatorname{Re}(\beta_q)}\big) t_{1}^{(k)}(q) \right) + \displaystyle\sum_{q=j}^{k-1} \overline{-(\alpha_q + \beta_q)} T^{(k)}(j, q)}{\beta_j + \overline{\beta}_k},\nonumber\\
t_2^{(k)} &= \frac{-\sqrt{-2\operatorname{Re}(\beta_k)} \left( 1 + \displaystyle\sum_{q=1}^{k-1} \big(-\sqrt{-2\operatorname{Re}(\beta_q)}\big) t_{1}^{(k)}(q) \right)}{\beta_k + \overline{\beta}_k},\label{t_anly}
\end{align}
for $j=1,\ldots,k-1$. A MATLAB function named \texttt{com\_t\_col} is provided in \cite{mycodes} to construct $t_1^{(k)}$ and $t_2^{(k)}$ from the shifts $\alpha_i$ and $\beta_i$.

It is shown in \cite{zulfiqar2026unified} that $(S_{\mathrm{fadi}}^{(k)},L_{\mathrm{fadi}}^{(k)})$ is related to $(S_{\mathrm{adi}}^{(k)},L_{\mathrm{adi}}^{(k)})$ as follows:
\begin{align}
S_{\mathrm{fadi}}^{(k)}=\big(T^{(k)}\big)^{-1}S_{\mathrm{adi}}^{(k)}T^{(k)}\quad \text{and} \quad L_{\mathrm{fadi}}^{(k)}=L_{\mathrm{adi}}^{(k)}T^{(k)}.
\end{align}
The matrices $\hat{A}^{(k)}$ and $\hat{B}^{(k)}$ are related to $S_{\mathrm{fadi}}^{(k)}$, $L_{\mathrm{fadi}}^{(k)}$, and $J_{\mathrm{fadi}}^{(k)}$ as follows:
\begin{align}
\hat{A}^{(k)}=-\Big(\big(J_{\mathrm{fadi}}^{(k)}\big)\big(\big(T^{(k)}\big)^{-1}S_{\mathrm{adi}}^{(k)}T^{(k)}\big)^*\big(J_{\mathrm{fadi}}^{(k)}\big)^{-1}\Big)\otimes I_m\quad \text{and}\quad \hat{B}^{(k)}=\Big(J_{\mathrm{fadi}}^{(k)}(L_{\mathrm{adi}}^{(k)}T^{(k)})^*\Big)\otimes I_m;
\end{align}
cf. \cite{zulfiqar2026unified}.

As shown in \cite{wolf2014h}, the following Lyapunov equation holds:
\[
-(S_{\mathrm{adi}}^{(k)})^*-S_{\mathrm{adi}}^{(k)}+(L_{\mathrm{adi}}^{(k)})^\top L_{\mathrm{adi}}^{(k)}=0.
\]
Note that the pairs $(\hat{A}^{(k)},\hat{B}^{(k)})$ and $\big(-(S_{\mathrm{adi}}^{(k)}\otimes I_m)^*,(L_{\mathrm{adi}}^{(k)}\otimes I_m)^\top\big)$ are related via the similarity transformation $\big(J_{\mathrm{fadi}}^{(k)}(T^{(k)})^*\big)^{-1}\otimes I_m$. Consequently,
\[
\hat{P}^{(k)}=Z^{(k)}(Z^{(k)})^*=\Big(\big(J_{\mathrm{fadi}}^{(k)}(T^{(k)})^*\big)\otimes I_m\Big)\Big(\big(J_{\mathrm{fadi}}^{(k)}(T^{(k)})^*\big)\otimes I_m\Big)^*
\]
due to the relationship between the controllability Gramians of two similar state-space realizations.

Therefore, the fADI-based low-rank approximation of $P$ is given by
\begin{align}
P\approx \tilde{P}=\Big(V_{\mathrm{fadi}}^{(k)}\big(J_{\mathrm{fadi}}^{(k)}(T^{(k)})^*\big)\otimes I_m\Big)\Big(V_{\mathrm{fadi}}^{(k)}\big(J_{\mathrm{fadi}}^{(k)}(T^{(k)})^*\big)\otimes I_m\Big)^*.\label{approx}
\end{align}
A MATLAB function named \texttt{compute\_CF} is provided in \cite{mycodes} to construct $Z^{(k)}$ from the shifts $\alpha_i$ and $\beta_i$.

It is shown in \cite{zulfiqar2026unified} that $V_{\mathrm{fadi}}^{(k)}$ solves the following Sylvester equation:
\begin{align}
AV_{\mathrm{fadi}}^{(k)}-EV_{\mathrm{fadi}}^{(k)}\hat{A}^{(k)}+B_{\perp,k}^{\mathrm{fadi}}L_{\mathrm{fadi}}^{(k)}&=0.\label{proj_sylv}
\end{align}
In \cite{wolf2014h}, a general residual expression is derived for projection-based approximation of the Lyapunov equation when the projection matrix solves a Sylvester equation of the form \eqref{proj_sylv}. Thus, we can use that residual expression here.

Define $F^{(k)}$ by
\[
F^{(k)}=EV_{\mathrm{fadi}}^{(k)}\Big(\big(J_{\mathrm{fadi}}^{(k)}\otimes I_m\big)-\hat{P}^{(k)}\Big)\big(L_{\mathrm{fadi}}^{(k)}\otimes I_m\big)^\top.
\]
Then the residual for the approximation \eqref{approx} is given by
\begin{align}
A\tilde{P}E^\top+E\tilde{P}A^\top +BB^\top = B_{\perp,k}^{\mathrm{fadi}}(B_{\perp,k}^{\mathrm{fadi}})^*+B_{\perp,k}^{\mathrm{fadi}}(F^{(k)})^*+F^{(k)}(B_{\perp,k}^{\mathrm{fadi}})^*.\label{res_eq}
\end{align}
It can readily be observed that when $\alpha_i=\overline{\beta}_i\in\mathbb{C}_{-}$, the following holds:
\begin{align}
B_{\perp,k}^{\mathrm{fadi}}&=B_{\perp,k}^{\mathrm{adi}},\quad J_{\mathrm{fadi}}^{(k)}=J_{\mathrm{fadi}}^{(k)}(T^{(k)})^*T^{(k)}(J_{\mathrm{fadi}}^{(k)})^*,\quad
V_{\mathrm{fadi}}^{(k)}\Big(\big((J_{\mathrm{fadi}}^{(k)})^{\frac{1}{2}}\otimes I_m\big)\Big)=Z_p^{(k)},\quad \text{and}\quad F^{(k)}=0.
\end{align}
Thus, the fADI-based approximation of $P$ is identical to that produced by the LRCF-ADI method when $\alpha_i=\overline{\beta}_i\in\mathbb{C}_{-}$. In other words, the poles $\beta_i$ of the matrix $\hat{A}^{(k)}$ in the LRCF-ADI method are conjugates of the ADI shifts $\alpha_i$, leading to the restriction that the ADI shifts $\alpha_i$ must have negative real parts, since this is necessary to ensure that $\hat{A}^{(k)}$ is Hurwitz and $\hat{P}^{(k)}$ has a unique solution. However, the fADI-based approximation of $P$ is much more flexible because the poles $\beta_i$ of $\hat{A}^{(k)}$ are not tied to the ADI shifts $\alpha_i$. Hence, only the $\beta_i$ need to have negative real parts to ensure that $\hat{A}^{(k)}$ is Hurwitz and $\hat{P}^{(k)}$ has a unique solution. The LRCF-ADI method for approximating $P$ is thus a special case of the fADI-based approximation of $P$, where the ADI shifts and the poles of $\hat{A}^{(k)}$ are coupled. While the pole-placement property of fADI was identified in \cite{zulfiqar2026unified}, it was not recognized that fADI can be used to approximate $P$ in the factorized form $P\approx Z^{(k)} (Z^{(k)})^*$, leading to the G-LRCF-ADI method.

Note that fADI has been used to approximate $P$ in the literature \cite{tu2009adi}. Shifts $\alpha_i$ with non-negative real parts have also been used \cite{tu2009adi}. However, those approximations focus on replacing $\beta_i$, $A^\top$, and $C^\top$ with $\overline{\alpha}_i$, $A$, and $B$, respectively, to approximate $P$ as $P\approx V_{\mathrm{fadi}}^{(k)}\big(J_{\mathrm{fadi}}^{(k)}\otimes I_m\big) (V_{\mathrm{fadi}}^{(k)})^*$. The issue with this approximation from a projection perspective is that $J_{\mathrm{fadi}}^{(k)}\otimes I_m$ is the unique positive-definite solution of the projected Lyapunov equation \eqref{proj_lyap} only when $\alpha_i\in\mathbb{C}_{-}$. Violating this condition means that $\hat{A}^{(k)}$ is not Hurwitz, so $J_{\mathrm{fadi}}^{(k)}$ is not positive-definite and thus cannot be decomposed into Cholesky-like factors. The approximation presented in this section requires that $\beta_i\in\mathbb{C}_{-}$, whereas the choice of $\alpha_i$ is not tied to $\beta_i$. Thus, in our approximation, $\big(J_{\mathrm{fadi}}^{(k)}(T^{(k)})^*\otimes I_m\big)\big(J_{\mathrm{fadi}}^{(k)}(T^{(k)})^*\otimes I_m\big)^*$ is always a unique positive-definite solution of the projected Lyapunov equation \eqref{proj_lyap} as long as $\alpha_i\neq-\beta_i$.

The pseudocode for G-LRCF-ADI is given in Algorithm \ref{alg1}. When $\alpha_i=\overline{\beta}_i\in\mathbb{C}_{-}$, G-LRCF-ADI reduces to the standard LRCF-ADI. A MATLAB function named \texttt{G\_LRCF\_ADI} that implements G-LRCF-ADI is available in \cite{mycodes}.
\begin{algorithm}[!h]
\caption{G-LRCF-ADI}\label{alg1}
\KwIn{Matrices of the Lyapunov equation \eqref{lyap_p}: $(E,A,B)$; ADI shifts: $\{\alpha_i\}_{i=1}^{k}\in\mathbb{C}$; Desired poles: $\{\beta_i\}_{i=1}^{k}\in\mathbb{C}_{-}$ satisfying $\beta_i\neq-\alpha_i$; Tolerance: $1>\tau>0$.}
\KwOut{Approximation: $P\approx \big(V_{\mathrm{fadi}}^{(k)}Z^{(k)}\big)\big(V_{\mathrm{fadi}}^{(k)}Z^{(k)}\big)^*$}

\textbf{Initialize:} $k=1$, $V_{\mathrm{fadi}}^{(0)}=[\;]$, $J^{(0)}=[\;]$, $T^{(0)}=[\;]$, $B_{\perp}=B$, and $F=[\;]$.

\While{$\frac{\|B_{\perp}B_{\perp}^*+B_{\perp}F^*+FB_{\perp}^*\|}{\|BB^\top\|}>\tau$}{
    Solve for $v_k$: $\big(A + \alpha_k E\big)v_k=B_{\perp}$.
    
    Set $\begin{bmatrix}t_1^{(k)}\\t_2^{(k)}\end{bmatrix}$ as in \eqref{t_anly}.
    
    Expand: $T^{(k)}=\begin{bmatrix}T^{(k-1)}&t_1^{(k)}\\0&t_2^{(k)}\end{bmatrix}$, $V_{\mathrm{fadi}}^{(k)}=\begin{bmatrix}V_{\mathrm{fadi}}^{(k-1)}&v_k\end{bmatrix}$, $J^{(k)}=\begin{bmatrix}J^{(k-1)}&0\\0&-(\alpha_k+\beta_k)I_m\end{bmatrix}$, and $Z^{(k)}=\begin{bmatrix}Z^{(k-1)}&0\\-\big((\alpha_k+\beta_k)(t_1^{(k)})^*\big)\otimes I_m&-(\alpha_k+\beta_k)(t_2^{(k)})^*I_m\end{bmatrix}$.
    
    Update: $B_{\perp} \gets B_{\perp} - (\alpha_k + \beta_k) E v_k$, $F\gets EV_{\mathrm{fadi}}^{(k)}\big(J^{(k)}-Z^{(k)}(Z^{(k)})^*\big)(-\mathbf{1}_{k}^\top\otimes I_m)^\top$, and
    $k\gets k+1$.
}
\end{algorithm}

In most practical situations, it is desirable to keep the matrices $V_{\mathrm{fadi}}^{(k)}$, $J_{\mathrm{fadi}}^{(k)}$, and $T^{(k)}$ real-valued. This can be achieved by grouping the shifts $\alpha_i$ and the desired poles $\beta_i$, and processing two shifts and two desired poles at a time when either the shifts or the desired poles are complex. The following grouping is suggested in \cite{benner2014computing} for realification:
\begin{enumerate}
  \item When $\alpha_i\in\mathbb{R}$ and $\beta_i\in\mathbb{R}$, no modification of fADI is required.
  \item When $\alpha_i\in\mathbb{C}$, $\alpha_{i+1}=\overline{\alpha}_i$, $\beta_i\in\mathbb{C}$, and $\beta_{i+1}=\overline{\beta}_i$,
  \begin{align}
 v_i^{\mathrm{fadi}}&=\begin{bmatrix}(A+\alpha_i E)^{-1}B_{\perp,i-1}^{\mathrm{fadi}}&(A+\overline{\alpha}_i E)^{-1}B_{\perp,i-1}^{\mathrm{fadi}}\end{bmatrix},\nonumber\\ s_{\mathrm{fadi}}^{(i)}&=\begin{bmatrix}-\operatorname{Re}(\beta_i)&\operatorname{Im}(\beta_i)\\-\operatorname{Im}(\beta_i)&-\operatorname{Re}(\beta_i)\end{bmatrix},\quad  l_{\mathrm{fadi}}^{(i)}=\begin{bmatrix}-1&0\end{bmatrix},\nonumber\\
 j_{\mathrm{fadi}}^{(i)}&=\begin{bmatrix} 
2g & \frac{\Delta + 2\operatorname{Im}(\beta_i)\delta}{\operatorname{Im}(\beta_i)} \\ 
\frac{\Delta - 2\operatorname{Im}(-\alpha_i)\delta}{\operatorname{Im}(-\alpha_i)} & \frac{g\big(2\operatorname{Im}(\beta_i)\operatorname{Im}(-\alpha_i) + \Delta\big)}{\operatorname{Im}(\beta_i)\operatorname{Im}(-\alpha_i)} 
\end{bmatrix},\label{jcase2}
  \end{align}
  where $g = \operatorname{Re}(-\alpha_i) - \operatorname{Re}(\beta_i)$, $\delta = \operatorname{Im}(-\alpha_i) - \operatorname{Im}(\beta_i)$, and $\Delta = |\alpha_i + \beta_i|^2$ \cite{benner2014computing,zulfiqar2026unified}.
  \item When $\alpha_i\in\mathbb{C}$, $\alpha_{i+1}=\overline{\alpha}_i$, $\beta_i\in\mathbb{R}$, and $\beta_{i+1}\in\mathbb{R}$,
   \begin{align}
 v_i^{\mathrm{fadi}}&=\begin{bmatrix}(A+\alpha_i E)^{-1}B_{\perp,i-1}^{\mathrm{fadi}}&(A+\overline{\alpha}_i E)^{-1}B_{\perp,i-1}^{\mathrm{fadi}}\end{bmatrix},\nonumber\\ s_{\mathrm{fadi}}^{(i)}&=\begin{bmatrix}-\beta_i&-1\\0&-\beta_{i+1}\end{bmatrix},\quad  l_{\mathrm{fadi}}^{(i)}=\begin{bmatrix}-1&0\end{bmatrix},\quad j_{\mathrm{fadi}}^{(i)}=\begin{bmatrix} 
g_1 + g_2 & -\Delta \\ 
\frac{g_1 g_2 - \operatorname{Im}(-\alpha_i)^2}{\operatorname{Im}(-\alpha_i)} & -\frac{g_1 \Delta}{\operatorname{Im}(-\alpha_i)}
\end{bmatrix},\label{jcase3}
  \end{align}
  where $g_1 = \operatorname{Re}(-\alpha_i) - \beta_i$, $g_2 = \operatorname{Re}(-\alpha_{i+1}) - \beta_{i+1}$, and $\Delta = \beta_{i+1}^2 - 2\operatorname{Re}(-\alpha_i)\beta_{i+1} + |-\alpha_i|^2$ \cite{benner2014computing,zulfiqar2026unified}.
  \item When $\alpha_i\in\mathbb{R}$, $\alpha_{i+1}\in\mathbb{R}$, $\beta_i\in\mathbb{C}$, and $\beta_{i+1}=\overline{\beta}_i$,
  \begin{align}
 v_i^{\mathrm{fadi}}&=\begin{bmatrix}(A+\alpha_i E)^{-1}B_{\perp,i-1}^{\mathrm{fadi}}&(A+\alpha_{i+1} E)^{-1}E(A+\alpha_i E)^{-1}B_{\perp,i-1}^{\mathrm{fadi}}\end{bmatrix},\nonumber\\ s_{\mathrm{fadi}}^{(i)}&=\begin{bmatrix}-\operatorname{Re}(\beta_i)&\operatorname{Im}(\beta_i)\\-\operatorname{Im}(\beta_i)&-\operatorname{Re}(\beta_i)\end{bmatrix},\quad  l_{\mathrm{fadi}}^{(i)}=\begin{bmatrix}-1&0\end{bmatrix},\quad j_{\mathrm{fadi}}^{(i)}=\begin{bmatrix} 
g_1 + g_2 & \frac{g_1 g_2 - \operatorname{Im}(\beta_i)^2}{\operatorname{Im}(\beta_i)} \\ 
\Delta & \frac{g_1 \Delta}{\operatorname{Im}(\beta_i)} 
\end{bmatrix},\label{jcase4}
  \end{align}
  where $g_1 = -\alpha_i - \operatorname{Re}(\beta_i)$, $g_2 = -\alpha_{i+1} - \operatorname{Re}(\beta_i)$, and $\Delta = \alpha_{i+1}^2 + 2\operatorname{Re}(\beta_i)\alpha_{i+1} + |\beta_i|^2$ \cite{benner2014computing,zulfiqar2026unified}.
\end{enumerate}

The corresponding real-valued projected matrices $\hat{A}^{(k)}=\hat{a}^{(k)}\otimes I_m$ and $\hat{B}^{(k)}=\hat{b}^{(k)}\otimes I_m$ are given by
\begin{align}
\hat{a}^{(k)}=\begin{bmatrix}\hat{a}^{(k-1)}&0\\\hat{a}_{21}^{(k)}&\hat{a}_{22}^{(k)}\end{bmatrix}\quad\text{and}\quad
\hat{b}^{(k)}=\begin{bmatrix}\hat{b}^{(k-1)}\\j_{\mathrm{fadi}}^{(k)}(l_{\mathrm{fadi}}^{(k)})^\top\end{bmatrix},
\end{align}
where $\hat{a}_{21}^{(k)}=-j_{\mathrm{fadi}}^{(k)}(l_{\mathrm{fadi}}^{(k)})^\top L_{\mathrm{fadi}}^{(k-1)}$ and $\hat{a}_{22}^{(k)}=-j_{\mathrm{fadi}}^{(k)}(s_{\mathrm{fadi}}^{(k)})^\top(j_{\mathrm{fadi}}^{(k)})^{-1}$; cf. \cite{zulfiqar2026unified}. A MATLAB function named \texttt{compute\_A\_B\_r} is provided in \cite{mycodes} to construct real-valued $(\hat{A}^{(k)},\hat{B}^{(k)})$ from the shifts $\alpha_i$ and $\beta_i$.

The real-valued $s_{\mathrm{adi}}^{(i)}$ and $l_{\mathrm{adi}}^{(i)}$ are given by
\begin{align}
s_{\mathrm{adi}}^{(i)} & =
\begin{cases} 
-\beta_i, \hspace*{5cm} \text{if } \operatorname{Im}(\beta_i) = 0, \\[6pt]
\gamma_{i}^2\begin{bmatrix} 
1 & \zeta_i\\ 
-\zeta_i & 0
\end{bmatrix}, \hspace*{3.5cm} \text{if } \operatorname{Im}(\beta_i) \neq 0,
\end{cases}\nonumber\\[6pt]
l_{\mathrm{adi}}^{(i)} &= 
\begin{cases} 
-\gamma_{i}, \hspace*{5cm} \text{if } \operatorname{Im}(\beta_i) = 0, \\[6pt]
-\sqrt{2}\gamma_{i}\begin{bmatrix} 
1 & 0
\end{bmatrix}, \hspace*{3.4cm} \text{if } \operatorname{Im}(\beta_i) \neq 0,
\end{cases}\label{sl_adi}
\end{align}
where $\gamma_{i}=\sqrt{-2\operatorname{Re}(\beta_i)}$, $\delta_{i}=\frac{\operatorname{Re}(\beta_i)}{\operatorname{Im}(\beta_i)}$, and $\zeta_i=\frac{\sqrt{1+\delta_i^2}}{2\delta_i}$; cf. \cite{zulfiqar2026unified}. A MATLAB function named \texttt{adi\_S\_L} is provided in \cite{mycodes} to construct real-valued $(S_{\mathrm{adi}}^{(k)}, L_{\mathrm{adi}}^{(k)})$.

Finally, the real-valued $T^{(k)}$ solves the Sylvester equation \eqref{trans_sylv} with real-valued $S_{\mathrm{adi}}^{(k)}$, $L_{\mathrm{adi}}^{(k)}$, $S_{\mathrm{fadi}}^{(k)}$, and $L_{\mathrm{fadi}}^{(k)}$; cf. \cite{zulfiqar2026unified}.  A MATLAB function named \texttt{fadi\_S\_L} is provided in \cite{mycodes} to construct real-valued $(S_{\mathrm{fadi}}^{(k)}, L_{\mathrm{fadi}}^{(k)})$. Unfortunately, the analytical formula \eqref{T_analy} for the complex-valued case is no longer valid in the real-valued case. Nevertheless, owing to the small size of $T^{(k)}$, the Sylvester equation \eqref{trans_sylv} can be solved cheaply.

The pseudocode for the proposed realified version of G-LRCF-ADI (G-LRCF-ADI$^{(\text{R})}$) is given in Algorithm $1^{(\text{R})}$. When $\alpha_i=\overline{\beta}_i\in\mathbb{C}_{-}$, G-LRCF-ADI$^{(\text{R})}$ reduces to the realified version of the standard LRCF-ADI \cite{benner2013reformulated}. A MATLAB function named \texttt{G\_LRCF\_ADI\_r} that implements G-LRCF-ADI$^{(\text{R})}$ is available in \cite{mycodes}.\\

\noindent\rule{\textwidth}{0.8pt}

\textbf{Algorithm 1$^{(\text{R})}$: G-LRCF-ADI$^{(\text{R})}$}

\noindent\rule{\textwidth}{0.8pt}

\textbf{Input:} Matrices of the Lyapunov equation \eqref{lyap_p}: $(E,A,B)$; ADI shifts: $\{\alpha_i\}_{i=1}^{k}\in\mathbb{C}$; Desired poles: $\{\beta_i\}_{i=1}^{k}\in\mathbb{C}_{-}$ satisfying $\beta_i\neq-\alpha_i$; Tolerance: $1>\tau>0$.

\textbf{Output:} Approximation: $P\approx \big(V_{\mathrm{fadi}}^{(k)}Z^{(k)}\big)\big(V_{\mathrm{fadi}}^{(k)}Z^{(k)}\big)^\top$.

1. \textbf{Initialize:} $k=1$, $V_{\mathrm{fadi}}^{(0)}=[\;]$, $J_{\mathrm{fadi}}^{(0)}=[\;]$, $T^{(0)}=[\;]$, $S_{\mathrm{adi}}^{(0)}=[\;]$, $L_{\mathrm{adi}}^{(0)}=[\;]$, $L_{\mathrm{fadi}}^{(0)}=[\;]$, $C_x^{(0)}=[\;]$, $B_{\perp}=B$, and $F=[\;]$.

2. \textbf{while} $\frac{\|B_{\perp}B_{\perp}^\top+B_{\perp}F^\top+FB_{\perp}^\top\|}{\|BB^\top\|}>\tau$ \textbf{do}

3. Solve $\big(A + \alpha_k E)v_k= B_{\perp}$ for $v_k$ and set $\gamma_k = \sqrt{-2\text{Re}(\beta_k)}$. 

4. \textbf{If} $\alpha_k\in\mathbb{R}$ and $\beta_k\in\mathbb{R}$ \textbf{do}

5. Set $j_{\mathrm{fadi}}^{(k)}=-(\alpha_k+\beta_k)$, $l_{\mathrm{fadi}}^{(k)}=-1$, $s_{\mathrm{adi}}^{(k)}=-\beta_k$, and $l_{\mathrm{adi}}^{(k)}=-\gamma_k$.

6. Expand $V_{\mathrm{fadi}}^{(k)} = \begin{bmatrix}V_{\mathrm{fadi}}^{(k-1)}&v_k\end{bmatrix}$, $J_{\mathrm{fadi}}^{(k)} = \mathrm{blkdiag}\big(J_{\mathrm{fadi}}^{(k-1)},j_{\mathrm{fadi}}^{(k)}\big)$, $S_{\mathrm{adi}}^{(k)}=\begin{bmatrix}S_{\mathrm{adi}}^{(k-1)}&(L_{\mathrm{adi}}^{(k-1)})^\top l_{\mathrm{adi}}^{(k)}\\0&s_{\mathrm{adi}}^{(k)}\end{bmatrix}$,\\
$L_{\mathrm{adi}}^{(k)}=\begin{bmatrix}L_{\mathrm{adi}}^{(k-1)}&l_{\mathrm{adi}}^{(k)}\end{bmatrix}$, and $L_{\mathrm{fadi}}^{(k)}=\begin{bmatrix}L_{\mathrm{fadi}}^{(k-1)}&l_{\mathrm{fadi}}^{(k)}\end{bmatrix}$.

7. Solve for $\begin{bmatrix}t_1^{(k)}\\t_2^{(k)}\end{bmatrix}$: $\Big(-(S_{\mathrm{adi}}^{(k)})^\top + \beta_k I\Big)\begin{bmatrix}t_1^{(k)}\\t_2^{(k)}\end{bmatrix} = \begin{bmatrix}(L_{\mathrm{adi}}^{(k-1)})^\top-T^{(k-1)} (C_{x}^{(k-1)})^\top\\l_{\mathrm{adi}}^{(k)}\end{bmatrix}$.

8. Expand $T^{(k)}=\begin{bmatrix}T^{(k-1)}&t_1^{(k)}\\0&t_2^{(k)}\end{bmatrix}$, $Z^{(k)}=\begin{bmatrix}Z^{(k-1)}&0\\\big(j_{\mathrm{fadi}}^{(k)}(t_1^{(k)})^\top\big)\otimes I_m&\big(j_{\mathrm{fadi}}^{(k)}(t_2^{(k)})^\top\big)\otimes I_m\end{bmatrix}$, and\\
$C_{x}^{(k)}=\begin{bmatrix}C_{x}^{(k-1)}&l_{\mathrm{fadi}}^{(k)}j_{\mathrm{fadi}}^{(k)}\end{bmatrix}$.

9. Update $B_{\perp} \gets B_{\perp} -(\alpha_k+\beta_k) E v_k$, $F\gets EV_{\mathrm{fadi}}^{(k)}\Big(\big(J_{\mathrm{fadi}}^{(k)}\otimes I_m\big)-Z^{(k)}(Z^{(k)})^\top\Big)\big(L_{\mathrm{fadi}}^{(k)}\otimes I_m\big)^\top$, and $k\gets k+1$.
                
10. \textbf{If} $\alpha_k\in\mathbb{C}$, $\alpha_{k+1}=\overline{\alpha}_k$, $\beta_k\in\mathbb{C}$, and $\beta_{k+1}=\overline{\beta}_k$ \textbf{do}

11. Set $j_{\mathrm{fadi}}^{(k)}$ and $l_{\mathrm{fadi}}^{(k)}$ as in \eqref{jcase2}, and $s_{\mathrm{adi}}^{(k)}$ and $l_{\mathrm{adi}}^{(k)}$ as in \eqref{sl_adi}.

12. Expand $V_{\mathrm{fadi}}^{(k)} = \begin{bmatrix}V_{\mathrm{fadi}}^{(k-1)}&\mathrm{Re}(v_k)&\mathrm{Im}(v_k)\end{bmatrix}$, $J_{\mathrm{fadi}}^{(k)} = \mathrm{blkdiag}\big(J_{\mathrm{fadi}}^{(k-1)},j_{\mathrm{fadi}}^{(k)}\big)$,\\
$S_{\mathrm{adi}}^{(k)}=\begin{bmatrix}S_{\mathrm{adi}}^{(k-1)}&(L_{\mathrm{adi}}^{(k-1)})^\top l_{\mathrm{adi}}^{(k)}\\0&s_{\mathrm{adi}}^{(k)}\end{bmatrix}$, $L_{\mathrm{adi}}^{(k)}=\begin{bmatrix}L_{\mathrm{adi}}^{(k-1)}&l_{\mathrm{adi}}^{(k)}\end{bmatrix}$, and $L_{\mathrm{fadi}}^{(k)}=\begin{bmatrix}L_{\mathrm{fadi}}^{(k-1)}&l_{\mathrm{fadi}}^{(k)}\end{bmatrix}$.
 
13. Solve $\Big(-(S_{\mathrm{adi}}^{(k)})^\top + \overline{\beta_k} I\Big)\begin{bmatrix}z_1^{(k)}\\z_2^{(k)}\end{bmatrix} = \begin{bmatrix}(L_{\mathrm{adi}}^{(k-1)})^\top-T^{(k-1)} (C_{x}^{(k-1)})^\top\\(l_{\mathrm{adi}}^{(k)})^\top\end{bmatrix}$ for $\begin{bmatrix}z_1^{(k)}\\z_2^{(k)}\end{bmatrix}$, and set $\begin{bmatrix}t_1^{(k)}\\t_2^{(k)}\end{bmatrix}=\begin{bmatrix}\mathrm{Re}(z_1^{(k)})&\mathrm{Im}(z_1^{(k)})\\\mathrm{Re}(z_2^{(k)})&\mathrm{Im}(z_2^{(k)})\end{bmatrix}$.

14. Expand $T^{(k)}=\begin{bmatrix}T^{(k-1)}&t_1^{(k)}\\0&t_2^{(k)}\end{bmatrix}$, $Z^{(k)}=\begin{bmatrix}Z^{(k-1)}&0\\\big(j_{\mathrm{fadi}}^{(k)}(t_1^{(k)})^\top\big)\otimes I_m&\big(j_{\mathrm{fadi}}^{(k)}(t_2^{(k)})^\top\big)\otimes I_m\end{bmatrix}$, and\\
$C_{x}^{(k)}=\begin{bmatrix}C_{x}^{(k-1)}&l_{\mathrm{fadi}}^{(k)}j_{\mathrm{fadi}}^{(k)}\end{bmatrix}$.

15. Update $B_{\perp} \gets B_{\perp} - E \begin{bmatrix}\mathrm{Re}(v_k)&\mathrm{Im}(v_k)\end{bmatrix}\big((j_{\mathrm{fadi}}^{(k)}(l_{\mathrm{fadi}}^{(k)})^\top)\otimes I_m\big)$, $F\gets EV_{\mathrm{fadi}}^{(k)}\Big(\big(J_{\mathrm{fadi}}^{(k)}\otimes I_m\big)-Z^{(k)}(Z^{(k)})^\top\Big)\big(L_{\mathrm{fadi}}^{(k)}\otimes I_m\big)^\top$, and $k\gets k+2$.

16. \textbf{If} $\alpha_k\in\mathbb{C}$, $\alpha_{k+1}=\overline{\alpha}_k$, $\beta_k\in\mathbb{R}$, and $\beta_{k+1}\in\mathbb{R}$ \textbf{do}

17. Set $j_{\mathrm{fadi}}^{(k)}$ and $l_{\mathrm{fadi}}^{(k)}$ as in \eqref{jcase3}, and $s_{\mathrm{adi}}^{(k)}$ and $l_{\mathrm{adi}}^{(k)}$ as in \eqref{sl_adi}.

18. Expand $V_{\mathrm{fadi}}^{(k)} = \begin{bmatrix}V_{\mathrm{fadi}}^{(k-1)}&\mathrm{Re}(v_k)&\mathrm{Im}(v_k)\end{bmatrix}$, $J_{\mathrm{fadi}}^{(k)} = \mathrm{blkdiag}\big(J_{\mathrm{fadi}}^{(k-1)},j_{\mathrm{fadi}}^{(k)}\big)$,\\
$S_{\mathrm{adi}}^{(k)}=\begin{bmatrix}S_{\mathrm{adi}}^{(k-1)}&(L_{\mathrm{adi}}^{(k-1)})^\top l_{\mathrm{adi}}^{(k)}\\0&s_{\mathrm{adi}}^{(k)}\end{bmatrix}$, $L_{\mathrm{adi}}^{(k)}=\begin{bmatrix}L_{\mathrm{adi}}^{(k-1)}&l_{\mathrm{adi}}^{(k)}\end{bmatrix}$, and $L_{\mathrm{fadi}}^{(k)}=\begin{bmatrix}L_{\mathrm{fadi}}^{(k-1)}&l_{\mathrm{fadi}}^{(k)}\end{bmatrix}$.

19. Solve for $\begin{bmatrix}z_1^{(k)}\\z_2^{(k)}\end{bmatrix}$: $\Big(-(S_{\mathrm{adi}}^{(k)})^\top + \beta_k I\Big)\begin{bmatrix}z_1^{(k)}\\z_2^{(k)}\end{bmatrix} = \begin{bmatrix}(L_{\mathrm{adi}}^{(k-1)})^\top-T^{(k-1)} (C_{x}^{(k-1)})^\top\\(l_{\mathrm{adi}}^{(k)})^\top\end{bmatrix}$.

20. Solve for $\begin{bmatrix}z_1^{(k+1)}\\z_2^{(k+1)}\end{bmatrix}$: $\Big(-(S_{\mathrm{adi}}^{(k)})^\top + \beta_{k+1} I\Big)\begin{bmatrix}z_1^{(k+1)}\\z_2^{(k+1)}\end{bmatrix} = -\begin{bmatrix}z_1^{(k)}\\z_2^{(k)}\end{bmatrix}$.

21. Set $\begin{bmatrix}t_1^{(k)}\\t_2^{(k)}\end{bmatrix}=\begin{bmatrix}z_1^{(k)}&z_1^{(k+1)}\\z_2^{(k)}&z_2^{(k+1)}\end{bmatrix}$, and expand $T^{(k)}=\begin{bmatrix}T^{(k-1)}&t_1^{(k)}\\0&t_2^{(k)}\end{bmatrix}$,\\
$Z^{(k)}=\begin{bmatrix}Z^{(k-1)}&0\\\big(j_{\mathrm{fadi}}^{(k)}(t_1^{(k)})^\top\big)\otimes I_m&\big(j_{\mathrm{fadi}}^{(k)}(t_2^{(k)})^\top\big)\otimes I_m\end{bmatrix}$, and $C_{x}^{(k)}=\begin{bmatrix}C_{x}^{(k-1)}&l_{\mathrm{fadi}}^{(k)}j_{\mathrm{fadi}}^{(k)}\end{bmatrix}$.
                        
22. Update $B_{\perp} \gets B_{\perp} - E \begin{bmatrix}\mathrm{Re}(v_k)&\mathrm{Im}(v_k)\end{bmatrix}\big((j_{\mathrm{fadi}}^{(k)}(l_{\mathrm{fadi}}^{(k)})^\top)\otimes I_m\big)$, $F\gets EV_{\mathrm{fadi}}^{(k)}\Big(\big(J_{\mathrm{fadi}}^{(k)}\otimes I_m\big)-Z^{(k)}(Z^{(k)})^\top\Big)\big(L_{\mathrm{fadi}}^{(k)}\otimes I_m\big)^\top$, and $k\gets k+2$.

23. \textbf{If} $\alpha_k\in\mathbb{R}$, $\alpha_{k+1}\in\mathbb{R}$, $\beta_k\in\mathbb{C}$, and $\beta_{k+1}=\overline{\beta}_k$ \textbf{do}

24. Solve for $v_{k+1}$: $\big(A + \alpha_{k+1} E)v_{k+1}= (E v_k)$.

25. Set $j_{\mathrm{fadi}}^{(k)}$ and $l_{\mathrm{fadi}}^{(k)}$ as in \eqref{jcase4}, and $s_{\mathrm{adi}}^{(k)}$ and $l_{\mathrm{adi}}^{(k)}$ as in \eqref{sl_adi}.

26. Expand $V_{\mathrm{fadi}}^{(k)} = \begin{bmatrix}V_{\mathrm{fadi}}^{(k-1)}&v_k&v_{k+1}\end{bmatrix}$, $J_{\mathrm{fadi}}^{(k)} = \mathrm{blkdiag}\big(J_{\mathrm{fadi}}^{(k-1)},j_{\mathrm{fadi}}^{(k)}\big)$, $S_{\mathrm{adi}}^{(k)}=\begin{bmatrix}S_{\mathrm{adi}}^{(k-1)}&(L_{\mathrm{adi}}^{(k-1)})^\top l_{\mathrm{adi}}^{(k)}\\0&s_{\mathrm{adi}}^{(k)}\end{bmatrix}$, $L_{\mathrm{adi}}^{(k)}=\begin{bmatrix}L_{\mathrm{adi}}^{(k-1)}&l_{\mathrm{adi}}^{(k)}\end{bmatrix}$, and $L_{\mathrm{fadi}}^{(k)}=\begin{bmatrix}L_{\mathrm{fadi}}^{(k-1)}&l_{\mathrm{fadi}}^{(k)}\end{bmatrix}$.

27. Solve $\Big(-(S_{\mathrm{adi}}^{(k)})^\top + \overline{\beta_k} I\Big)\begin{bmatrix}z_1^{(k)}\\z_2^{(k)}\end{bmatrix} = \begin{bmatrix}(L_{\mathrm{adi}}^{(k-1)})^\top-T^{(k-1)} (C_{x}^{(k-1)})^\top\\(l_{\mathrm{adi}}^{(k)})^\top\end{bmatrix}$ for $\begin{bmatrix}z_1^{(k)}\\z_2^{(k)}\end{bmatrix}$, and set $\begin{bmatrix}t_1^{(k)}\\t_2^{(k)}\end{bmatrix}=\begin{bmatrix}\mathrm{Re}(z_1^{(k)})&\mathrm{Im}(z_1^{(k)})\\\mathrm{Re}(z_2^{(k)})&\mathrm{Im}(z_2^{(k)})\end{bmatrix}$.

28. Expand $T^{(k)}=\begin{bmatrix}T^{(k-1)}&t_1^{(k)}\\0&t_2^{(k)}\end{bmatrix}$, $Z^{(k)}=\begin{bmatrix}Z^{(k-1)}&0\\\big(j_{\mathrm{fadi}}^{(k)}(t_1^{(k)})^\top\big)\otimes I_m&\big(j_{\mathrm{fadi}}^{(k)}(t_2^{(k)})^\top\big)\otimes I_m\end{bmatrix}$, and $C_{x}^{(k)}=\begin{bmatrix}C_{x}^{(k-1)}&l_{\mathrm{fadi}}^{(k)}j_{\mathrm{fadi}}^{(k)}\end{bmatrix}$.

29. Update $B_{\perp} \gets B_{\perp} - E \begin{bmatrix}v_k&v_{k+1}\end{bmatrix}\big((j_{\mathrm{fadi}}^{(k)}(l_{\mathrm{fadi}}^{(k)})^\top)\otimes I_m\big)$, $F\gets EV_{\mathrm{fadi}}^{(k)}\Big(\big(J_{\mathrm{fadi}}^{(k)}\otimes I_m\big)-Z^{(k)}(Z^{(k)})^\top\Big)$ $\big(L_{\mathrm{fadi}}^{(k)}\otimes I_m\big)^\top$, and $k\gets k+2$.

\noindent\rule{\textwidth}{0.8pt}
\subsection{Extension to the frequency-limited Lyapunov equation\label{FLBT}}
Frequency-limited balanced truncation (FLBT) is a model order reduction (MOR) algorithm that produces a reduced-order model (ROM) approximating $G(s)$ with improved accuracy over a desired frequency interval $\Omega:=[\omega_1,\omega_2]\cup[-\omega_2,-\omega_1]$ rad/s \cite{gawronski1990model}. The frequency-limited controllability Gramian $P_\Omega$ solves the following Lyapunov equation:
\begin{align}
AP_\Omega E^\top+EP_\Omega A^\top+EF_\Omega(E,A)BB^\top+BB^\top F_\Omega(E,A)^\top E^\top&=0,
\end{align}
where 
\begin{align}
F_\Omega(E,A)=\frac{1}{2\pi}\int_{\Omega}(j\nu E-A)^{-1}d\nu.\label{mat_log_int}
\end{align}
The Gramian $P_\Omega$ can be expressed in integral form as
\begin{align}
P_\Omega&=\frac{1}{2\pi}\int_{\Omega}(j\nu E-A)^{-1}BB^\top(-j\nu E^\top-A^\top)^{-1}d\nu.\label{gram_int}
\end{align}
The basic principle of numerical integration is to replace the integrand by an interpolant at selected nodes and then integrate the interpolant instead of the original integrand. Since the ADI shifts in G-LRCF-ADI$^{(\text{R})}$ can lie on the imaginary axis, we can use G-LRCF-ADI$^{(\text{R})}$ to approximate the integrals \eqref{mat_log_int} and \eqref{gram_int}, thereby recursively approximating $F_\Omega(E,A)$ and $P_\Omega$, respectively.

Assume that all ADI shifts $\alpha_i=j\omega_i$ lie within the desired frequency interval $\Omega$. Then the integrand in \eqref{gram_int} can be replaced by the interpolant constructed by G-LRCF-ADI$^{(\text{R})}$, yielding the approximation
\begin{align}
P_\Omega\approx \tilde{P}_\Omega&= V_{\mathrm{fadi}}^{(k)}\Big(\frac{1}{2\pi}\int_{\Omega}(j\nu I-\hat{A}^{(k)})^{-1}\hat{B}^{(k)}(\hat{B}^{(k)})^\top(-j\nu I-(\hat{A}^{(k)})^\top)^{-1}d\nu\Big)(V_{\mathrm{fadi}}^{(k)})^\top\nonumber\\
&=V_{\mathrm{fadi}}^{(k)}\Big(F_\Omega(I,\hat{A}^{(k)})Z^{(k)}(Z^{(k)})^\top+Z^{(k)}(Z^{(k)})^\top(F_\Omega(I,\hat{A}^{(k)}))^\top\Big)(V_{\mathrm{fadi}}^{(k)})^\top\nonumber\\
&=V_{\mathrm{fadi}}^{(k)}\Big(\big(F_\Omega(I,\hat{a}^{(k)})J_{\mathrm{fadi}}^{(k)}(T^{(k)})^\top T^{(k)}(J_{\mathrm{fadi}}^{(k)})^\top\nonumber\\
&\hspace*{3cm}+J_{\mathrm{fadi}}^{(k)}(T^{(k)})^\top T^{(k)}(J_{\mathrm{fadi}}^{(k)})^\top(F_\Omega(I,\hat{a}^{(k)}))^\top\big)\otimes I_m\Big)(V_{\mathrm{fadi}}^{(k)})^\top.\nonumber
\end{align}
The matrix $F_\Omega(I,\hat{a}^{(k)})$ can be computed as
\[
F_\Omega(I,\hat{a}^{(k)})=\operatorname{Re}\big(F(\hat{a}^{(k)},\omega_2)-F(\hat{a}^{(k)},\omega_1)\big),
\]
where
\[
F(\hat{a}^{(k)},\omega)=\frac{j}{2\pi}\ln\big((j\omega I+\hat{a}^{(k)})(-j\omega I+\hat{a}^{(k)})^{-1}\big);
\]
cf. \cite{petersson2014model}.
Owing to the lower triangular structure of $\hat{a}^{(k)}=\begin{bmatrix}\hat{a}^{(k-1)}&0\\\hat{a}_{21}^{(k)}&\hat{a}_{22}^{(k)}\end{bmatrix}$, $F(\hat{a}^{(k)},\omega)$ also has a lower triangular structure, i.e.,
\[
F(\hat{a}^{(k)},\omega)=\begin{bmatrix}F(\hat{a}^{(k-1)},\omega)&0\\f_{21}^{(k)}&F(\hat{a}_{22}^{(k)},\omega)\end{bmatrix};
\]
cf. \cite{petersson2014model}.

Furthermore, due to the commutativity
\[
\hat{a}^{(k)}F(\hat{a}^{(k)},\omega)=F(\hat{a}^{(k)},\omega)\hat{a}^{(k)},
\]
the off-diagonal block of $F(\hat{a}^{(k)},\omega)$ can be computed by solving the following Sylvester equation:
\begin{align}
\hat{a}_{22}^{(k)}f_{21}^{(k)}-f_{21}^{(k)}\hat{a}^{(k-1)}+\hat{a}_{21}^{(k)}F(\hat{a}^{(k-1)},\omega)-F(\hat{a}_{22}^{(k)},\omega)\hat{a}_{21}^{(k)}=0;
\end{align}
cf. \cite{zulfiqar2022adaptive}. This Sylvester equation can be solved cheaply owing to the small sizes of the matrices $\hat{a}^{(k-1)}$ and $\hat{a}_{22}^{(k)}$. Furthermore, since $F(\hat{a}_{22}^{(k)},\omega)$ can be computed cheaply, $F_\Omega(I,\hat{a}^{(k)})$ can be computed recursively and efficiently. A MATLAB function named \texttt{compute\_F\_w} is provided in \cite{mycodes} to construct $F_\Omega(I,\hat{a}^{(k)})$ from the shifts $\alpha_i$ and $\beta_i$.

Consequently, G-LRCF-ADI$^{(\mathrm{R})}$ can be used to recursively approximate $P_\Omega$ according to the principle of numerical integration. Similarly, the integral \eqref{mat_log_int} can be approximated using the same principle, giving
\[
F_\Omega(E,A)B\approx V_{\mathrm{fadi}}^{(k)}F_\Omega(I,\hat{A}^{(k)})\hat{B}^{(k)}.
\]

Since $(\hat{A}^{(k)},\hat{B}^{(k)})$ is guaranteed to be controllable, its frequency-limited controllability Gramian
\[
\hat{P}_\Omega^{(k)}=F_\Omega(I,\hat{A}^{(k)})Z^{(k)}(Z^{(k)})^\top+Z^{(k)}(Z^{(k)})^\top(F_\Omega(I,\hat{A}^{(k)}))^\top
\]
is guaranteed to be positive definite \cite{gawronski1990model}. Thus, $\hat{P}_\Omega^{(k)}$ can be factorized as $\hat{P}_\Omega^{(k)}=Z_\Omega^{(k)} (Z_\Omega^{(k)})^\top$, which is useful for FLBT because the square-root algorithm \cite{tombs1987truncated} requires the Gramians in this form.

By exploiting the low-rank structure of $\tilde{P}_\Omega$ and the small number of columns of $B$, the normalized residual
\begin{align}
\frac{\|A\tilde{P}_\Omega E^\top+E\tilde{P}_\Omega A^\top+EF_\Omega(E,A)BB^\top+BB^\top F_\Omega(E,A)^\top E^\top\|_F}{\|EF_\Omega(E,A)BB^\top+BB^\top F_\Omega(E,A)^\top E^\top\|_F}\label{fl_res}
\end{align}
can be evaluated efficiently without forming dense $n \times n$ matrices.

The residual can be factored as
\begin{align}
  &A\tilde{P}_\Omega E^\top+E\tilde{P}_\Omega A^\top+EF_\Omega(E,A)BB^\top+BB^\top F_\Omega(E,A)^\top E^\top\nonumber\\
  &=\begin{bmatrix}AV_{\mathrm{fadi}}^{(k)}\hat{P}_\Omega^{(k)}&EV_{\mathrm{fadi}}^{(k)}&EF_\Omega(E,A)B&B\end{bmatrix}
  \begin{bmatrix}EV_{\mathrm{fadi}}^{(k)}&AV_{\mathrm{fadi}}^{(k)}\hat{P}_\Omega^{(k)}&B&EF_\Omega(E,A)B\end{bmatrix}^\top.\label{fl_n_res}
\end{align}
Using the identity $\|R\|_F^2 = \operatorname{trace}(R^\top R)$ and the cyclic permutation property of the trace, the normalized residual can be computed efficiently from the factorization \eqref{fl_n_res}. However, this computation requires explicitly forming the product $F_\Omega(E,A)B$, which is not feasible in large-scale settings. A possible stopping criterion is based on stagnation of the singular values of $(V_{\mathrm{fadi}}^{(k)})^\top E V_{\mathrm{fadi}}^{(k)}$, which indicates that the new shifts are not adding meaningful information to the interpolation basis.
\subsection{Extension to the time-limited Lyapunov equation\label{TLBT}}
Time-limited balanced truncation (TLBT) is a MOR algorithm that produces a ROM approximating $G(s)$ with improved accuracy over a desired time interval $[0,\tau]$ s \cite{gawronski1990model}. The time-limited controllability Gramian $P_\tau$ solves the following Lyapunov equation:
\begin{align}
AP_\tau E^\top + E P_\tau A^\top + BB^\top  -e^{AE^{-1}\tau} BB^\top  e^{E^{-\top}A^\top\tau}=0.
\end{align}
Since G-LRCF-ADI$^{(\text{R})}$ is an interpolatory projection algorithm, the matrix exponential product $e^{E^{-1}A\tau}E^{-1} B$ can be approximated via projection as
\begin{align}
e^{E^{-1}A\tau} E^{-1}B\approx V_{\mathrm{fadi}}^{(k)} e^{\hat{A}^{(k)}\tau} \hat{B}^{(k)}=V_{\mathrm{fadi}}^{(k)} \big(e^{\hat{a}^{(k)}\tau}\otimes I_m\big) \hat{B}^{(k)}.
\end{align}
Owing to the lower triangular structure of $\hat{a}^{(k)}=\begin{bmatrix}\hat{a}^{(k-1)}&0\\\hat{a}_{21}^{(k)}&\hat{a}_{22}^{(k)}\end{bmatrix}$, the matrix exponential $e^{\hat{a}^{(k)}\tau}$ also has a lower triangular structure, i.e.,
\[
e^{\hat{a}^{(k)}\tau}=\begin{bmatrix}e^{\hat{a}^{(k-1)}\tau}&0\\e_{21}^{(k)}&e^{\hat{a}_{22}^{(k)}\tau}\end{bmatrix}.
\]
Furthermore, due to the commutativity
\[
\hat{a}^{(k)}e^{\hat{a}^{(k)}\tau}=e^{\hat{a}^{(k)}\tau}\hat{a}^{(k)},
\]
the off-diagonal block of $e^{\hat{a}^{(k)}\tau}$ can be computed by solving the following Sylvester equation:
\begin{align}
\hat{a}_{22}^{(k)}e_{21}^{(k)}-e_{21}^{(k)}\hat{a}^{(k-1)}+\hat{a}_{21}^{(k)}e^{\hat{a}^{(k-1)}\tau}-e^{\hat{a}_{22}^{(k)}\tau}\hat{a}_{21}^{(k)}=0.
\end{align}
This Sylvester equation can be solved cheaply owing to the small sizes of the matrices $\hat{a}^{(k-1)}$ and $\hat{a}_{22}^{(k)}$. Furthermore, since $e^{\hat{a}_{22}^{(k)}\tau}$ can be computed cheaply, $e^{\hat{a}^{(k)}\tau}$ can be computed recursively and efficiently. A MATLAB function named \texttt{compute\_exp\_A} is provided in \cite{mycodes} to construct $e^{\hat{a}^{(k)}\tau}$ from the shifts $\alpha_i$ and $\beta_i$.

The matrix $P_\tau$ can be expressed in integral form as
\begin{align}
P_\tau=\int_{0}^{\tau}e^{E^{-1}At}E^{-1}BB^\top E^{-\top} e^{A^\top E^{-\top}t}dt.
\end{align}
Replacing the matrix exponential product $e^{E^{-1}At}E^{-1} B$ with its G-LRCF-ADI$^{(\text{R})}$-based approximation yields
\begin{align}
P_\tau&\approx \tilde{P}_\tau=V_{\mathrm{fadi}}^{(k)}\Big(\int_{0}^{\tau}\big(e^{\hat{a}^{(k)}t}\otimes I_m\big) \hat{B}^{(k)}(\hat{B}^{(k)})^\top \big(e^{\hat{a}^{(k)}t}\otimes I_m\big)^\top dt\Big)(V_{\mathrm{fadi}}^{(k)})^\top\nonumber\\
&=V_{\mathrm{fadi}}^{(k)}\Big(Z^{(k)}(Z^{(k)})^\top-\big(e^{\hat{a}^{(k)}\tau}\otimes I_m\big) Z^{(k)}(Z^{(k)})^\top \big(e^{\hat{a}^{(k)}\tau}\otimes I_m\big)^\top\Big)(V_{\mathrm{fadi}}^{(k)})^\top.\nonumber
\end{align}
Consequently, G-LRCF-ADI$^{(\mathrm{R})}$ can be used to recursively approximate $P_\tau$. Furthermore, since $(\hat{A}^{(k)},\hat{B}^{(k)})$ is guaranteed to be controllable, its time-limited controllability Gramian
\[
\hat{P}_\tau^{(k)}=Z^{(k)}(Z^{(k)})^\top-e^{\hat{A}^{(k)}\tau}Z^{(k)}(Z^{(k)})^\top e^{(\hat{A}^{(k)})^\top\tau}
\]
is guaranteed to be positive definite. Thus, $\hat{P}_\tau^{(k)}$ can be factorized as $\hat{P}_\tau^{(k)}=Z_\tau^{(k)} (Z_\tau^{(k)})^\top$, which is useful for TLBT because the square-root algorithm \cite{tombs1987truncated} requires the Gramians in this form.

By exploiting the low-rank structure of $\tilde{P}_\tau$ and the small number of columns of $B$, the normalized residual
\begin{align}
\frac{\|A\tilde{P}_\tau E^\top + E \tilde{P}_\tau A^\top + BB^\top  -e^{AE^{-1}\tau} BB^\top  e^{E^{-\top}A^\top\tau}\|_F}{\|BB^\top  -e^{AE^{-1}\tau} BB^\top  e^{E^{-\top}A^\top\tau}\|_F}\label{tl_res}
\end{align}
can be evaluated efficiently without forming dense $n \times n$ matrices.

The residual can be factored as
\begin{align}
  &A\tilde{P}_\tau E^\top + E \tilde{P}_\tau A^\top + BB^\top  -e^{AE^{-1}\tau} BB^\top  e^{E^{-\top}A^\top\tau}\nonumber\\
  &=\begin{bmatrix}AV_{\mathrm{fadi}}^{(k)}\hat{P}_\tau^{(k)}&EV_{\mathrm{fadi}}^{(k)}&B&-e^{AE^{-1}\tau} B\end{bmatrix}
  \begin{bmatrix}EV_{\mathrm{fadi}}^{(k)}&AV_{\mathrm{fadi}}^{(k)}\hat{P}_\tau^{(k)}&B&e^{AE^{-1}\tau} B\end{bmatrix}^\top.\label{tl_n_res}
\end{align}

The computation of the residual for this approximation requires explicitly computing the matrix exponential product $e^{AE^{-1}\tau} B$, which is not feasible in large-scale settings. A possible stopping criterion is based on stagnation of the singular values of $(V_{\mathrm{fadi}}^{(k)})^\top E V_{\mathrm{fadi}}^{(k)}$, which indicates that the new shifts are not adding meaningful information to the interpolation basis.
\subsection{Extension to the Riccati equation\label{LQGBT}}
G-LRCF-ADI$^{(\mathrm{R})}$ can also be used to compute a low-rank solution of a Riccati equation. Consider the following Riccati equation:
\begin{align}
AP_{\mathrm{ricc}}E^\top+EP_{\mathrm{ricc}}A^\top +BB^\top -EP_{\mathrm{ricc}}C^\top CP_{\mathrm{ricc}}E^\top&=0.
\end{align}
Let $\hat{P}_{\mathrm{ricc}}^{(k)}$ solve the following projected Riccati equation:
\begin{align}
\hat{A}^{(k)}\hat{P}_{\mathrm{ricc}}^{(k)}+\hat{P}_{\mathrm{ricc}}^{(k)}(\hat{A}^{(k)})^\top +\hat{B}^{(k)}(\hat{B}^{(k)})^\top -\hat{P}_{\mathrm{ricc}}^{(k)}(V_{\mathrm{fadi}}^{(k)})^\top C^\top CV_{\mathrm{fadi}}^{(k)}\hat{P}_{\mathrm{ricc}}^{(k)}&=0.\label{proj_ricc}
\end{align}
Since the projected pair $(\hat{A}^{(k)},\hat{B}^{(k)})$ is guaranteed to be stable and controllable in G-LRCF-ADI$^{(\text{R})}$, there exists a positive-definite stabilizing solution to the projected Riccati equation \eqref{proj_ricc} such that the matrix
\[
\hat{A}^{(k)}-\hat{P}_{\mathrm{ricc}}^{(k)}(V_{\mathrm{fadi}}^{(k)})^\top C^\top CV_{\mathrm{fadi}}^{(k)}
\]
is Hurwitz.

Due to the quadratic term $\hat{P}_{\mathrm{ricc}}^{(k)}(V_{\mathrm{fadi}}^{(k)})^\top C^\top CV_{\mathrm{fadi}}^{(k)}\hat{P}_{\mathrm{ricc}}^{(k)}$, the stabilizing solution to the projected Riccati equation \eqref{proj_ricc} cannot be updated recursively unless $\alpha_i=\overline{\beta}_i$, in which case the approximation $V_{\mathrm{fadi}}^{(k)}\hat{P}_{\mathrm{ricc}}^{(k)}(V_{\mathrm{fadi}}^{(k)})^\top$ is identical to that produced by the ADI method for Riccati equations (RADI) \cite{benner2018radi}; see \cite{zulfiqar2026unified} for details. When $\alpha_i\neq\overline{\beta}_i$, the projected Riccati equation \eqref{proj_ricc} must be solved afresh in each iteration, which is a disadvantage compared with RADI \cite{benner2018radi}. Thus, in the case of Riccati equations, the flexibility in shift selection and prescribed poles comes with a trade-off.

The next proposition shows that the residual can be computed as in the Lyapunov equation case.
\begin{proposition}
Define $F_{\mathrm{ricc}}^{(k)}$ by
\[
F_{\mathrm{ricc}}^{(k)}=EV_{\mathrm{fadi}}^{(k)}\big(\hat{B}^{(k)}-\hat{P}_{\mathrm{ricc}}^{(k)}(L_{\mathrm{fadi}}^{(k)})^\top\big).
\]
Then the residual $R_{\mathrm{ricc}}^{(k)}$ for the approximation $P_{\mathrm{ricc}}\approx \tilde{P}_{\mathrm{ricc}}^{(k)}= V_{\mathrm{fadi}}^{(k)}\hat{P}_{\mathrm{ricc}}^{(k)}(V_{\mathrm{fadi}}^{(k)})^\top$, defined by
\begin{align}
R_{\mathrm{ricc}}^{(k)}=A\tilde{P}_{\mathrm{ricc}}^{(k)}E^\top+E\tilde{P}_{\mathrm{ricc}}^{(k)}A^\top +BB^\top -E\tilde{P}_{\mathrm{ricc}}^{(k)}C^\top C\tilde{P}_{\mathrm{ricc}}^{(k)}E^\top
\end{align}
is given by
\[
R_{\mathrm{ricc}}^{(k)}=B_{\perp,k}^{\mathrm{fadi}}(B_{\perp,k}^{\mathrm{fadi}})^\top+B_{\perp,k}^{\mathrm{fadi}}(F_{\mathrm{ricc}}^{(k)})^\top+F_{\mathrm{ricc}}^{(k)}(B_{\perp,k}^{\mathrm{fadi}})^\top.
\]
\end{proposition}
\begin{proof}
The proof is given in Appendix A.
\end{proof}
\section{Applications to data-driven modeling}
It is shown in \cite{zulfiqar2026unified} that fADI \cite{tu2009adi,benner2014computing} is essentially a rational interpolation algorithm. However, a direct connection between rational Krylov subspaces and fADI has not been established. In this section, we establish a direct analytical relation between Krylov-subspace-based rational interpolation and fADI, which is useful for data-driven reduced-order modeling. Using this relation, we show that G-LRCF-ADI can be used as a system identification tool, since ROMs with important system properties can be constructed from measured frequency-domain data without any knowledge of the state-space realization or transfer function of the original model.

Define $V_{\mathrm{kry}}^{(k)}$ by
\[
V_{\mathrm{kry}}^{(k)}=\begin{bmatrix}(-\alpha_1E-A)^{-1}B&\cdots&(-\alpha_kE-A)^{-1}B\end{bmatrix}.
\]
The next proposition establishes the relation between $V_{\mathrm{fadi}}^{(k)}$ and $V_{\mathrm{kry}}^{(k)}$.
\begin{proposition}
Let the ADI shifts $\alpha_i$ be distinct for $i=1,\ldots,k$, and let $\beta_i$ be the prescribed poles of $\hat{G}^{(k)}(s)$. Define the upper triangular matrix $T_v^{(k)}$ by
\begin{equation}
\label{eq:T_closed_form}
    T_v^{(k)}(j,i) = 
    \begin{cases} 
    (-1)^{i} \dfrac{\prod_{q=1}^{i-1}(\alpha_j + \beta_q)}{\prod_{\substack{q=1 \\ q \neq j}}^{i}(\alpha_q - \alpha_j)}, & 1 \le j \le i, \\[10pt]
    0, & j > i.
    \end{cases}
\end{equation}
Then the fADI matrix $V_{\mathrm{fadi}}^{(k)}$ is related to $V_{\mathrm{kry}}^{(k)}$ by
\[
V_{\mathrm{fadi}}^{(k)}=V_{\mathrm{kry}}^{(k)}(T_v^{(k)}\otimes I_m).
\]
\end{proposition}
\begin{proof}
The proof is given in Appendix B.
\end{proof}
A MATLAB function named \texttt{construct\_T} is provided in \cite{mycodes} to construct $T_v^{(k)}$ from the shifts $\alpha_i$ and $\beta_i$.

In \cite{zulfiqar2026data,zulfiqar2026proj}, various data-driven modeling algorithms were presented based on the pole-placement property of fADI. However, the exact relation between $V_{\mathrm{kry}}^{(k)}$ and $V_{\mathrm{fadi}}^{(k)}$ was not recognized. Consequently, $(T_v^{(k)})^{-1}$ implicitly appeared in all the algorithms presented in \cite{zulfiqar2026data,zulfiqar2026proj}. As the number of data samples increased, the dimension of $(T_v^{(k)})^{-1}$ implicitly present in the state-space realizations of \cite{zulfiqar2026data,zulfiqar2026proj} grew, causing numerical issues. In most numerical algorithms, it is a rule of thumb to avoid explicit inverses as much as possible for numerical stability. Now that an analytical expression for $T_v^{(k)}$ and the exact relation between $V_{\mathrm{kry}}^{(k)}$ and $V_{\mathrm{fadi}}^{(k)}$ have been identified, we revisit the data-driven models presented in \cite{zulfiqar2026data,zulfiqar2026proj} and replace the matrix $CV_{\mathrm{kry}}^{(k)}$ with $CV_{\mathrm{fadi}}^{(k)}$ to avoid the implicit appearance of $(T_v^{(k)})^{-1}$. The state-space realization presented in \cite{zulfiqar2026data,zulfiqar2026proj} can be viewed as being obtained by applying the similarity transformation
\[
\hat{A}^{(k)}\gets \big(T_v^{(k)}\otimes I_m\big)\hat{A}^{(k)}\big((T_v^{(k)})^{-1}\otimes I_m\big),\quad \hat{B}^{(k)}\gets\big(T_v^{(k)}\otimes I_m\big)\hat{B}^{(k)},\quad \hat{C}^{(k)}\gets CV_{\mathrm{fadi}}^{(k)}\big((T_v^{(k)})^{-1}\otimes I_m\big),
\]
which was numerically unstable due to the implicit presence of $(T_v^{(k)})^{-1}$, not recognized in \cite{zulfiqar2026data,zulfiqar2026proj}.

Dually, define $W_{\mathrm{kry}}^{(k)}$ by
\[
W_{\mathrm{kry}}^{(k)}=\begin{bmatrix}(-\overline{\mu}_1E^\top-A^\top)^{-1}C^\top&\cdots&(-\overline{\mu}_kE^\top-A^\top)^{-1}C^\top\end{bmatrix}.
\]
Next, construct $T_w^{(k)}$ from the distinct ADI shifts $\mu_i$ and the prescribed poles $\nu_i$ by
\begin{equation}
    T_w^{(k)}(j,i) = 
    \begin{cases} 
    (-1)^{i} \dfrac{\prod_{q=1}^{i-1}(\overline{\mu}_j + \overline{\nu}_q)}{\prod_{\substack{q=1 \\ q \neq j}}^{i}(\overline{\mu}_q - \overline{\mu}_j)}, & 1 \le j \le i, \\[10pt]
    0, & j > i.
    \end{cases}
\end{equation}
Then the matrices $Z_{\mathrm{fadi}}^{(k)}$ and $W_{\mathrm{kry}}^{(k)}$ are related by
\[
Z_{\mathrm{fadi}}^{(k)}=W_{\mathrm{kry}}^{(k)}(T_w^{(k)}\otimes I_p).
\]
\subsection{Alternative to the Loewner framework \cite{mayo2007framework}}
Write $G(s)$ as
\[
G(s)=H(s)+D=H(s)+G(\infty),
\]
since $D=\lim_{s\rightarrow\infty}G(s)$.

In the Loewner framework \cite{mayo2007framework}, the matrices $\bar{E}^{(k)}$, $\bar{A}^{(k)}$, $\bar{B}^{(k)}$, $\bar{C}^{(k)}$, and $\bar{D}$ are constructed from the samples $H(-\alpha_i)$, $H(-\mu_i)$, and $G(\infty)$:
\begin{align}
\bar{E}^{(k)}(j,i)&=(W_{\mathrm{kry}}^{(k)}(:,i))^*EV_{\mathrm{kry}}^{(k)}(:,j)=\frac{H(-\mu_i)-H(-\alpha_j)}{\mu_i-\alpha_j},\nonumber\\
\bar{A}^{(k)}(j,i)&=(W_{\mathrm{kry}}^{(k)}(:,i))^*AV_{\mathrm{kry}}^{(k)}(:,j)=\frac{\alpha_jH(-\alpha_j)-\mu_iH(-\mu_i)}{\mu_i-\alpha_j},\nonumber\\
\bar{B}^{(k)}(j,:)&=(W_{\mathrm{kry}}^{(k)}(:,j))^*B=H(-\mu_j),\nonumber\\
\bar{C}^{(k)}(:,i)&=CV_{\mathrm{kry}}^{(k)}(:,i)=H(-\alpha_i),\nonumber\\
\bar{D}&=G(\infty).\nonumber
\end{align}
This realization satisfies the following interpolation conditions:
\[
G(-\alpha_i)=\bar{C}^{(k)}(-\alpha_iI-\bar{A}^{(k)})^{-1}\bar{B}^{(k)}+\bar{D}\quad \text{and}\quad G(-\mu_j)=\bar{C}^{(k)}(-\mu_jI-\bar{A}^{(k)})^{-1}\bar{B}^{(k)}+\bar{D}.
\]

Define the following state-space matrices, constructed from the samples of $H(-\alpha_i)$, $H(-\mu_i)$, and $G(\infty)$:
\begin{align}
\hat{E}_{\mathrm{fadi}}^{(k)}&=(Z_{\mathrm{fadi}}^{(k)})^*EV_{\mathrm{fadi}}^{(k)}=(T_w^{(k)}\otimes I_p)^*\bar{E}^{(k)}(T_v^{(k)}\otimes I_m),\nonumber\\
\hat{A}_{\mathrm{fadi}}^{(k)}&=(Z_{\mathrm{fadi}}^{(k)})^*AV_{\mathrm{fadi}}^{(k)}=(T_w^{(k)}\otimes I_p)^*\bar{A}^{(k)}(T_v^{(k)}\otimes I_m),\nonumber\\
\hat{B}_{\mathrm{fadi}}^{(k)}&=(Z_{\mathrm{fadi}}^{(k)})^*B=(T_w^{(k)}\otimes I_p)^*\bar{B}^{(k)},\nonumber\\
\hat{C}_{\mathrm{fadi}}^{(k)}&=CV_{\mathrm{fadi}}^{(k)}=\bar{C}^{(k)}(T_v^{(k)}\otimes I_m),\nonumber\\
\hat{D}_{\mathrm{fadi}}&=G(\infty).\nonumber
\end{align}
Note that as the number of samples $k$ increases, the computation of $T_v^{(k)}$ and $T_w^{(k)}$ becomes numerically unstable, since the products that appear in their analytical formulas may cause overflow. Therefore, we do not recommend first constructing $(\bar{A}^{(k)},\bar{B}^{(k)},\bar{C}^{(k)},\bar{D},\bar{E}^{(k)})$, then constructing $T_v^{(k)}$ and $T_w^{(k)}$, and then multiplying to obtain $(\hat{A}_{\mathrm{fadi}}^{(k)},\hat{B}_{\mathrm{fadi}}^{(k)},\hat{C}_{\mathrm{fadi}}^{(k)},\hat{D}_{\mathrm{fadi}},\hat{E}_{\mathrm{fadi}}^{(k)})$. Instead, $(\hat{A}_{\mathrm{fadi}}^{(k)},\hat{B}_{\mathrm{fadi}}^{(k)},\hat{C}_{\mathrm{fadi}}^{(k)},\hat{D}_{\mathrm{fadi}},\hat{E}_{\mathrm{fadi}}^{(k)})$ should be constructed recursively to inherit the numerical stability of the G-LRCF-ADI algorithm. We refer to the recursive construction of $(\hat{A}_{\mathrm{fadi}}^{(k)},\hat{B}_{\mathrm{fadi}}^{(k)},\hat{C}_{\mathrm{fadi}}^{(k)},\hat{D}_{\mathrm{fadi}},\hat{E}_{\mathrm{fadi}}^{(k)})$ directly from the samples $H(-\alpha_i)$, $H(-\mu_i)$, and $G(\infty)$ as the ``\textit{Data-driven ADI framework}''. The pseudocode for the Data-driven ADI framework is given in Algorithm \ref{alg_adi_fram}. A MATLAB function named \texttt{compute\_adi\_rom} that implements the Data-driven ADI framework is available in \cite{mycodes}. Although the interpolant constructed using Algorithm \ref{alg_adi_fram} is merely a different state-space realization of the same interpolant constructed by the Loewner framework \cite{mayo2007framework}, it has several interesting and useful properties and applications, which are explored in the sequel.
\begin{algorithm}[!h]
\SetAlgoLined
\caption{Data-driven ADI framework}
\label{alg_adi_fram}

\KwIn{Samples $H(-\alpha_i)$ and $H(-\mu_i)$ for $i=1,\dots,k$; shifts $\alpha_i, \beta_i, \mu_i, \nu_i$ for $i=1,\dots,k$.}
\KwOut{Matrices of the Interpolant: $\hat{E}_{\mathrm{fadi}}^{(k)}, \hat{A}_{\mathrm{fadi}}^{(k)} \in \mathbb{C}^{pk \times mk}$, $\hat{B}_{\mathrm{fadi}}^{(k)} \in \mathbb{C}^{pk \times m}$, $\hat{C}_{\mathrm{fadi}}^{(k)} \in \mathbb{C}^{p \times mk}$.}

Initialize $r_i = -H(-\mu_i)$ for $i=1,\dots,k$.

\For{$j = 1, \dots, k$}{
    $\hat{B}_{\mathrm{fadi}}^{(k)}(j,:) = r_j$.
    
    \For{$i = j+1, \dots, k$}{
        $r_i \leftarrow r_i - \frac{\mu_j + \nu_j}{\mu_j - \mu_i} (r_i - r_j)$.
    }
}

Initialize $s_j = -H(-\alpha_j)$ for $j=1,\dots,k$.

\For{$i = 1, \dots, k$}{
    $\hat{C}_{\mathrm{fadi}}^{(k)}(:,i) = s_i$.
    
    \For{$j = i+1, \dots, k$}{
        $s_j \leftarrow s_j - \frac{\alpha_i + \beta_i}{\alpha_j - \alpha_i} (s_i - s_j)$.
    }
}

Initialize $S_i = \hat{C}_{\mathrm{fadi}}^{(k)}(:,i)$ for $i=1,\dots,k$.

\For{$j = 1, \dots, k$}{
    $R_j \leftarrow \hat{B}_{\mathrm{fadi}}^{(k)}(j,:)$.
    
    \For{$i = 1, \dots, k$}{
        $\hat{E}_{\mathrm{fadi}}^{(k)}(j,i) = \frac{R_j - S_i}{\alpha_i - \mu_j}$, $\hat{A}_{\mathrm{fadi}}^{(k)}(j,i) = S_i - \mu_j \hat{E}_{\mathrm{fadi}}^{(k)}(j,i)$, $R_j \leftarrow R_j - (\alpha_i + \beta_i) \hat{E}_{\mathrm{fadi}}^{(k)}(j,i)$, $S_i \leftarrow S_i - (\mu_j + \nu_j) \hat{E}_{\mathrm{fadi}}^{(k)}(j,i)$.
    }
}
\end{algorithm}
\subsection{Data-driven rational interpolation with prescribed poles}
This subsection presents an elegant alternative, which is theoretically equivalent to the methods in \cite{zulfiqar2026data,zulfiqar2026proj}, for performing data-driven rational interpolation with prescribed poles. The interpolant $\hat{G}^{(k)}(s)=\hat{C}^{(k)}(sI-\hat{A}^{(k)})^{-1}\hat{B}^{(k)}+\hat{D}$ with prescribed poles $\beta_i$ can be constructed from the samples $G(-\alpha_i)$ and $G(\infty)$ as follows:
\begin{align}
\hat{A}^{(k)}&=\begin{bmatrix}\hat{A}^{(k-1)}&0\\\big((\alpha_k+\beta_k)\mathbf{1}_{k-1}^\top\big)\otimes I_m&\beta_kI_m\end{bmatrix},\quad \hat{B}^{(k)}=\begin{bmatrix}\hat{B}^{(k-1)}\\(\alpha_k+\beta_k)I_m\end{bmatrix},\nonumber\\
 \hat{C}^{(k)}&=\hat{C}_{\mathrm{fadi}}^{(k)}-\big(\hat{D}-G(\infty)\big)(-\mathbf{1}_{k}^\top\otimes I_m).\label{int1_data}
\end{align}
This realization satisfies $G(-\alpha_i)=\hat{G}^{(k)}(-\alpha_i)$. The matrix $\hat{D}$ in \eqref{int1_data} may be treated as a free parameter. In most cases, $\hat{D}$ is set equal to $G(\infty)$; however, in some situations (such as zero placement, discussed in the sequel), it is advantageous to set $\hat{D}\neq G(\infty)$. The prescribed poles $\beta_i$ can be placed anywhere in the complex plane as long as the condition $\alpha_i\neq-\beta_i$ holds. In some applications, such as small-signal stability analysis in power systems, it is desirable to place the poles in the right half of the $s$-plane to capture the oscillatory behavior associated with unstable modes. If all the ADI shifts $\alpha_i$ are purely imaginary, i.e., $\alpha_i=j\omega_i$, the samples $G(-\alpha_i)$ can be measured in an experimental setting. The formula \eqref{int1_data} provides a novel way to pack the measured data into a state-space realization with user-specified pole locations.

A dual state-space realization can also be constructed. Let $\mu_i$ be ADI shifts and $\nu_i$ be desired poles satisfying $\mu_i\neq-\nu_i$. Then the state-space realization
\begin{align}
\hat{A}^{(k)}&=\begin{bmatrix}\hat{A}^{(k-1)}&\big((\mu_k+\nu_k)\mathbf{1}_{k-1}\big)\otimes I_p\\0&\nu_kI_p\end{bmatrix},\quad \hat{B}^{(k)}=\hat{B}_{\mathrm{fadi}}^{(k)}-(-\mathbf{1}_{k}\otimes I_p)\big(\hat{D}-G(\infty)\big),\nonumber\\ \hat{C}^{(k)}&=\begin{bmatrix}\hat{C}^{(k-1)}&(\mu_k+\nu_k)I_p\end{bmatrix}\label{int2_data}
\end{align}
satisfies the interpolation conditions
\[
G(-\mu_i)=\hat{G}^{(k)}(-\mu_i).
\]
A MATLAB function named \texttt{ddripp} is provided in \cite{mycodes} to construct the state-space realizations discussed in this subsection from the samples $G(-\alpha_i)$, $G(-\mu_i)$, and $G(\infty)$.
\subsection{Data-driven rational interpolation with prescribed zeros\label{SWBT}}
This subsection presents an elegant, theoretically equivalent alternative to the methods in \cite{zulfiqar2026data,zulfiqar2026proj} for performing data-driven rational interpolation with prescribed zeros. Assume that the model in \eqref{int1_data} satisfies the interpolation conditions $G(-\alpha_i)=\hat{G}^{(k)}(-\alpha_i)$ and that $\hat{D}$ is invertible. Define $\hat{C}_w^{(k)}$ and $\hat{A}_w^{(k)}$ by
\[
\hat{C}_w^{(k)}=\hat{D}^{-1}\hat{C}^{(k)}\quad \text{and}\quad \hat{A}_w^{(k)}=\hat{A}^{(k)}-\hat{B}^{(k)}\hat{C}_w^{(k)}.
\]
Next, define
\[
\hat{H}_w^{(k)}(s)=\hat{C}_w^{(k)}\big(sI-\hat{A}_w^{(k)}\big)^{-1}\hat{B}^{(k)}.
\]
Then, construct $\hat{C}_z^{(k)}$ using Algorithm \ref{alg_adi_fram}, replacing $H(-\alpha_i)$ with $\hat{H}_w^{(k)}(-\alpha_i)$.

Next, update $\hat{A}^{(k)}$ and $\hat{C}^{(k)}$ by
\[
\hat{A}^{(k)}=\hat{A}^{(k)}+\hat{B}^{(k)}\hat{C}_z^{(k)}\quad \text{and}\quad \hat{C}^{(k)}=\hat{D}\hat{C}_z^{(k)}.
\]
Then $\hat{G}^{(k)}(s)=\hat{C}^{(k)}(sI-\hat{A}^{(k)})^{-1}\hat{B}^{(k)}+\hat{D}$ has transmission zeros at $\beta_i$ and satisfies the interpolation conditions $G(-\alpha_i)=\hat{G}^{(k)}(-\alpha_i)$; see \cite{zulfiqar2026data,zulfiqar2026proj} for details.

A dual realization with transmission zeros at $\nu_i$, satisfying the interpolation conditions $G(-\mu_i)=\hat{G}^{(k)}(-\mu_i)$, can also be computed by constructing the interpolant as in \eqref{int2_data} instead of \eqref{int1_data}. A MATLAB function named \texttt{ddripz} is provided in \cite{mycodes} to construct the state-space realizations discussed in this subsection from the samples $G(-\alpha_i)$, $G(-\mu_i)$, and $G(\infty)$.
\subsection{Data-driven balanced truncation}
The data-driven low-rank balanced truncation (BT) algorithm presented in this subsection is theoretically equivalent to that presented in \cite{zulfiqar2026data,zulfiqar2026proj}; however, it avoids solving the Sylvester equations and projected Lyapunov equations required by the algorithms in \cite{zulfiqar2026data,zulfiqar2026proj}. Instead, it uses analytical formulas to construct the approximate truncated balanced realization directly from samples of the original transfer function $G(s)$.

Recall that $P$ can be approximated using G-LRCF-ADI as
\[
P\approx \big(V_{\mathrm{fadi}}^{(k)}Z^{(k)}\big)\big(V_{\mathrm{fadi}}^{(k)}Z^{(k)}\big)^*.
\]
Dually, $Q$ can be approximated using G-LRCF-ADI as
\[
Q\approx \big(Z_{\mathrm{fadi}}^{(k)}Y^{(k)}\big)\big(Z_{\mathrm{fadi}}^{(k)}Y^{(k)}\big)^*.
\]
Note that $Z^{(k)}$ and $Y^{(k)}$ are constructed non-intrusively from the analytical formula \eqref{T_analy} using $(\alpha_i,\beta_i)$ and $(\mu_i,\nu_i)$, respectively; this is implemented in the MATLAB function \texttt{compute\_CF} available in \cite{mycodes}.

Next, construct $(\hat{A}_{\mathrm{fadi}}^{(k)},\hat{B}_{\mathrm{fadi}}^{(k)},\hat{C}_{\mathrm{fadi}}^{(k)},\hat{D}_{\mathrm{fadi}},\hat{E}_{\mathrm{fadi}}^{(k)})$ from the samples $H(-\alpha_i)$, $H(-\mu_i)$, and $G(\infty)$ using Algorithm \ref{alg_adi_fram}.

Then compute the following singular value decomposition:
\begin{align} 
(Y^{(k)})^* \hat{E}_{\mathrm{fadi}}^{(k)} Z^{(k)} = \begin{bmatrix} U_1 & U_2 \end{bmatrix} \begin{bmatrix} \Sigma_r & 0 \\ 0 & \Sigma_{n-r} \end{bmatrix} \begin{bmatrix} V_1^* \\ V_2^* \end{bmatrix}. \label{bsa_svd} 
\end{align}
Next, compute the following projection matrices:
\begin{align}
W_r = Y^{(k)} U_1 \Sigma_r^{-1/2} \quad \text{and} \quad V_r = Z^{(k)} V_1 \Sigma_r^{-1/2}.  \label{bsa_proj}
\end{align}
Then the $r$th-order ROM based on G-LRCF-ADI BT can be obtained as
\begin{align}  
A_r = W_r^* \hat{A}_{\mathrm{fadi}}^{(k)} V_r,\quad  B_r = W_r^* \hat{B}_{\mathrm{fadi}}^{(k)},\quad  C_r = \hat{C}_{\mathrm{fadi}}^{(k)} V_r,\quad D_r=G(\infty). \label{bsa_rom}
\end{align}
Note that this ROM can be constructed non-intrusively from the samples $G(-\alpha_i)$, $G(-\mu_i)$, and $G(\infty)$ without accessing the state-space realization $(A,B,C,D,E)$. When all $\alpha_i$ and $\mu_i$ lie on the imaginary axis, these samples can be measured in an experimental setting. A MATLAB function named \texttt{ddbt} is provided in \cite{mycodes} to perform data-driven BT using the samples $G(-\alpha_i)$, $G(-\mu_i)$, and $G(\infty)$.

Along similar lines, data-driven implementations of FLBT, TLBT, linear quadratic Gaussian BT \cite{jonckheere1983new}, and self-weighted BT \cite{zhou1995frequency} can be obtained using the results of subsections \ref{FLBT}, \ref{TLBT}, \ref{LQGBT}, and \ref{SWBT}, respectively; these are omitted here for brevity. Such implementations are theoretically equivalent to those presented in \cite{zulfiqar2026data,zulfiqar2026proj}; however, they avoid the Sylvester equations and inverses required by the algorithms in \cite{zulfiqar2026data,zulfiqar2026proj}, which cause numerical issues. 
\section{Numerical Results}
In this section, the numerical performance of G-LRCF-ADI$^{(\text{R})}$ and the Data-driven ADI framework is evaluated using benchmark dynamical systems available in the literature. The MATLAB codes to reproduce the results in this section are publicly available in \cite{mycodes}. All tests are performed using MATLAB R2025b on a Windows 11 laptop equipped with 32 GB of random access memory (RAM) and an Intel(R) Core(TM) Ultra 9 285H 2.9 GHz processor.

\subsection{Low-rank solutions of Lyapunov equations}
\subsubsection{Standard Lyapunov equation}
The first benchmark is a semi-discretized heat transfer problem for optimal cooling of steel profiles \cite{benner2005semi}, also known as the rail model. The dimensions of the matrices are as follows: \(E \in \mathbb{R}^{317{,}377 \times 317{,}377}\), \(A \in \mathbb{R}^{317{,}377 \times 317{,}377}\), and \(B \in \mathbb{R}^{317{,}377 \times 7}\). The ADI shifts used in the LRCF-ADI algorithm are generated automatically using the shift generation strategy proposed in \cite{zulfiqar2026unified}. When the normalized residual $\frac{\|A\tilde{P}E^\top+E\tilde{P}A^\top +BB^\top\|_2}{\|BB^\top\|_2}$ drops below the tolerance \(10^{-10}\), the algorithms are considered to have converged. As discussed earlier, G-LRCF-ADI$^{(\text{R})}$ and LRCF-ADI should theoretically perform identically when $\beta_i=\overline{\alpha}_i$. To verify this numerically, we take the shifts $\alpha_i$ used by LRCF-ADI and set $\beta_i=\overline{\alpha}_i$ in G-LRCF-ADI$^{(\text{R})}$. The decay of the normalized residual is plotted in Figure \ref{fig1}. It can be seen that there are no substantial round-off errors, and both LRCF-ADI and G-LRCF-ADI$^{(\text{R})}$ perform identically.

\begin{figure}[!h]
  \centering
  \includegraphics[width=12cm]{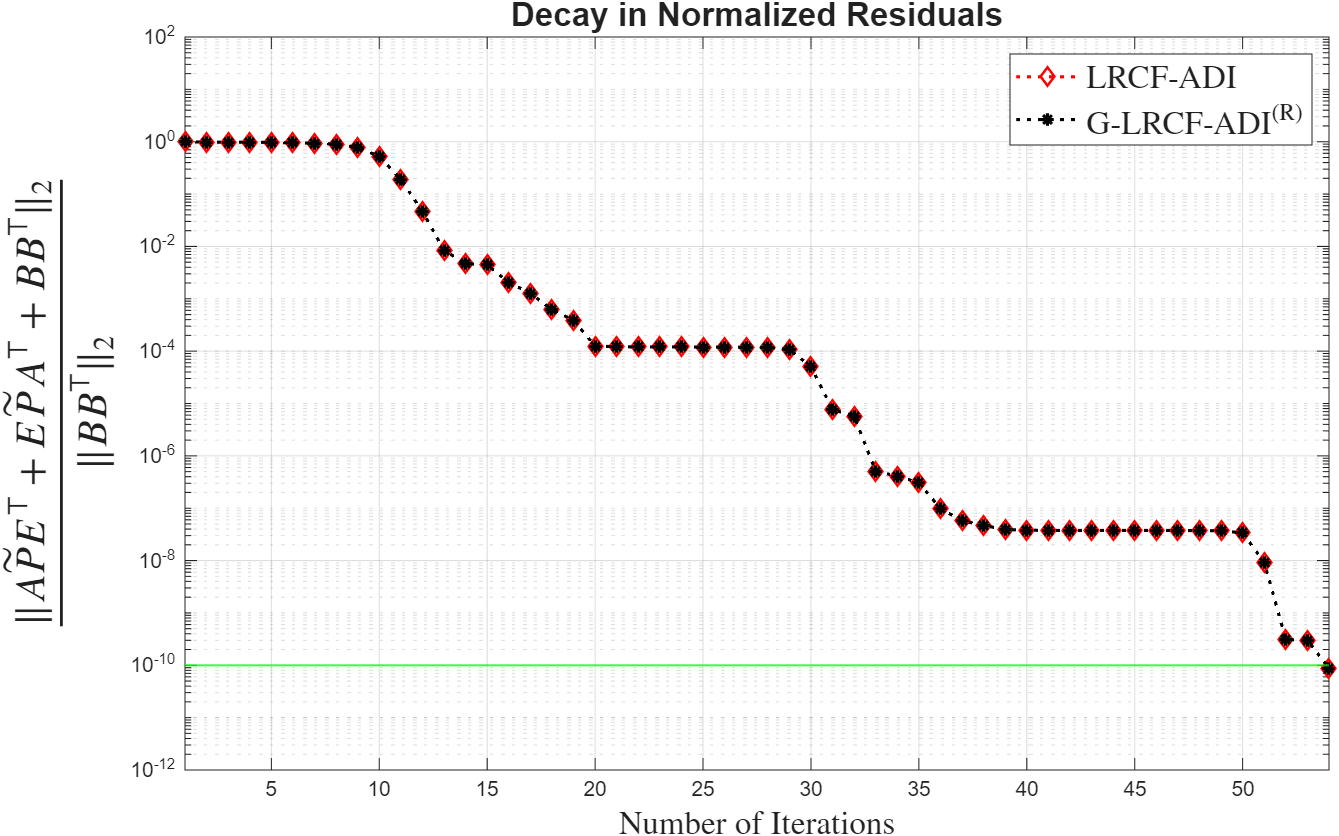}
  \caption{Decay of the normalized residual}\label{fig1}
\end{figure}

Theoretically, if we reverse the order of the prescribed poles $\beta_i$, the final approximation produced by LRCF-ADI and G-LRCF-ADI$^{(\text{R})}$ should be identical. To verify this numerically, we use the MATLAB command \texttt{fliplr} to reverse the sequence of prescribed poles $\beta_i$ in G-LRCF-ADI$^{(\text{R})}$. The decay of the normalized residual is plotted in Figure \ref{fig2}. It can be seen that round-off errors cause overflows, preventing G-LRCF-ADI$^{(\text{R})}$ from eventually converging to LRCF-ADI.

\begin{figure}[!h]
  \centering
  \includegraphics[width=12cm]{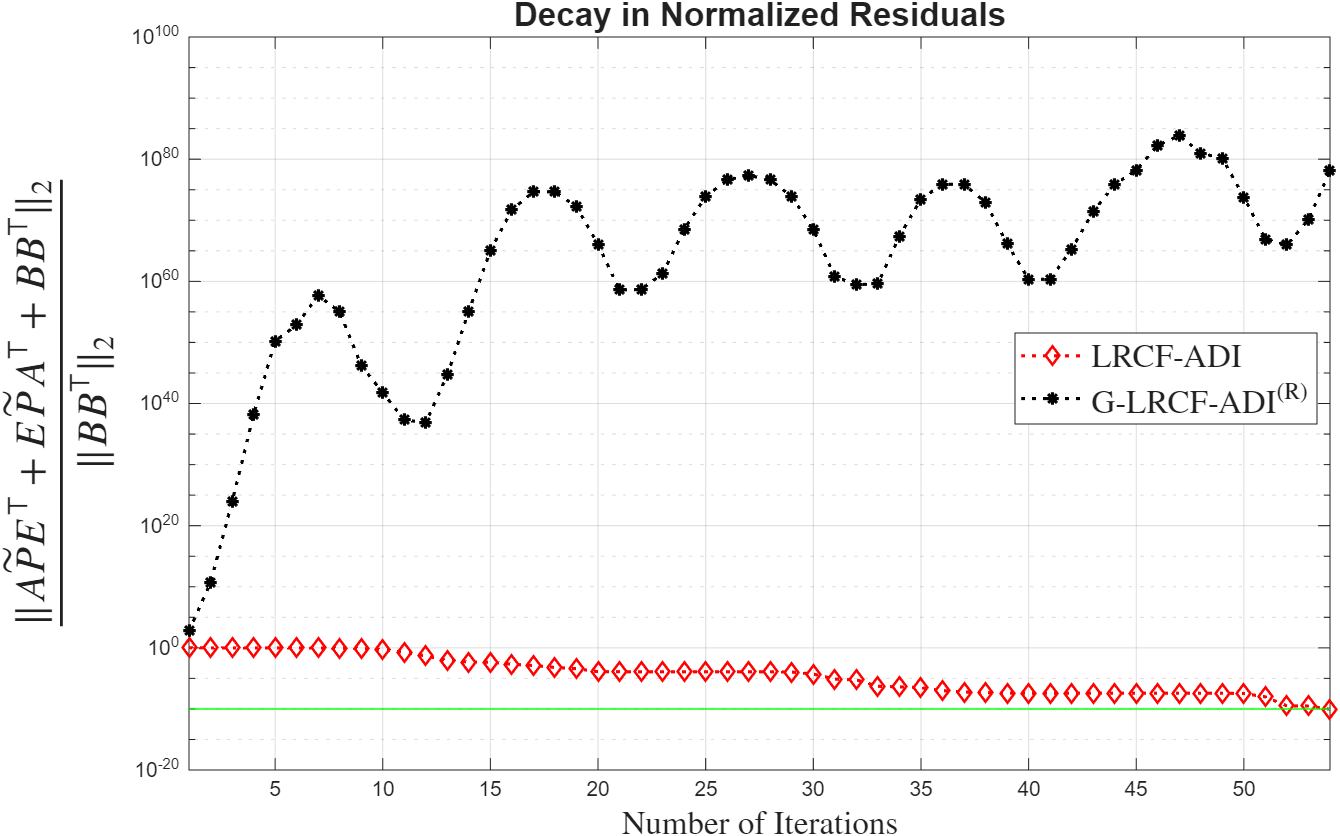}
  \caption{Decay of the normalized residual}\label{fig2}
\end{figure}

Next, to check whether G-LRCF-ADI$^{(\text{R})}$ eventually converges to LRCF-ADI in the final iteration, we limit the maximum number of iterations to $10$ so that round-off error does not accumulate excessively. It can be seen from Figure \ref{fig3} that, with a modest number of iterations, G-LRCF-ADI$^{(\text{R})}$ eventually converges to LRCF-ADI, in accordance with theoretical expectations.

\begin{figure}[!h]
  \centering
  \includegraphics[width=12cm]{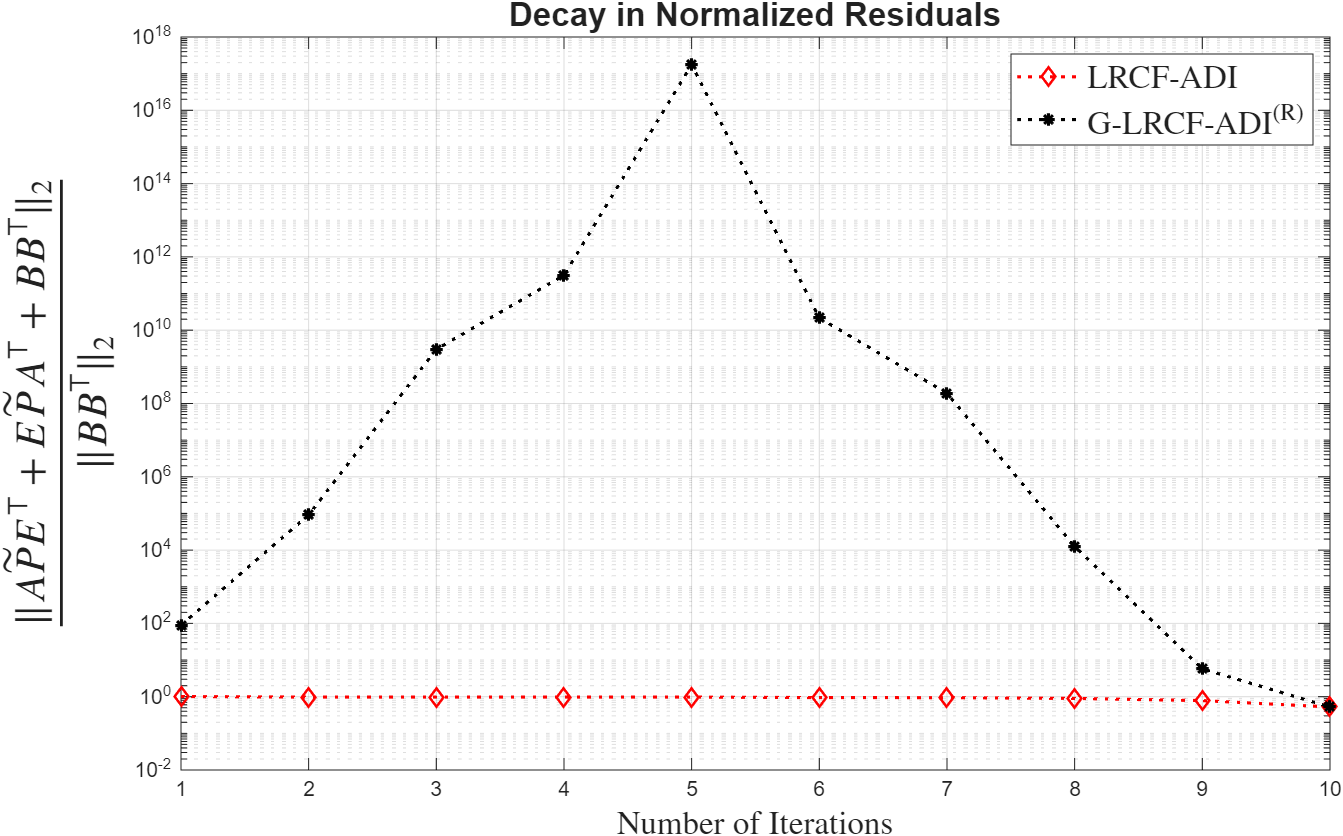}
  \caption{Decay of the normalized residual}\label{fig3}
\end{figure}
\subsubsection{Frequency- and time-limited Lyapunov equations}
This benchmark generalizes the well-known Penzl procedure for constructing a dynamical system with three peaks in the frequency response plot \cite{morPen06,chahlaoui2005benchmark}. The state-space matrices of this benchmark, of order \(n=10^5\), are as follows. The system has peaks in its frequency response at \(w_1=10\) rad/s, \(w_2=15\) rad/s, and \(w_3=20\) rad/s. Define \(A_i = \begin{bsmallmatrix} -1 & w_i \\ -w_i & -1 \end{bsmallmatrix}\) for \(i=1,2,3\), and
\(A_4 = -\mathrm{diag}(1,\dots,n{-}6)\).
The system matrices are
\begin{equation}
E=I,\quad
A = \mathrm{blkdiag}(A_1,A_2,A_3,A_4) \in \mathbb{R}^{n\times n},\quad
B = \begin{bmatrix} 10\,\mathbf{1}_6 \\ 10^{-4}\mathbf{1}_{n-6} \end{bmatrix},\quad
C = \begin{bmatrix} 10\,\mathbf{1}_6^\top & 10^{-4}\mathbf{1}_{n-6}^\top \end{bmatrix}.\label{triple_peak}
\end{equation}
Since \(A\) is block diagonal, \(F_\Omega(E,A)\) and \(e^{AE^{-1}\tau}\) are also block diagonal and can therefore be computed despite the large order \(n=10^5\). The desired frequency interval for FLBT is set to \([5,25]\) rad/s to capture the three peaks associated with the six dominant poles \(-1\pm j10\), \(-1\pm j15\), and \(-1\pm j20\). The 200 ADI shifts \(\alpha_i=\pm j\omega_i\) and the 200 ADI shifts \(\mu_i=\pm j\omega_i\) are purely imaginary, with imaginary parts chosen as exponentially spaced frequency points in \([5,25]\) rad/s. The prescribed poles are set to \(\beta_i=-1+\alpha_i\) and \(\nu_i=-1+\mu_i\). The relative errors obtained with G-LRCF-ADI\(^{(\text{R})}\) for approximating \(F_\Omega(E,A)B\) and \(F_\Omega(E,A)^\top C^\top\) are given by
\begin{align}
\frac{\|F_\Omega(E,A)B-V_{\mathrm{fadi}}^{(k)}F_{\Omega}(I,\hat{A}^{(k)})\hat{B}^{(k)}\|_F}{\|F_\Omega(E,A)B\|_F}&=8.6597\times 10^{-9},\nonumber\\
\frac{\|F_\Omega(E,A)^\top C^\top-Z_{\mathrm{fadi}}^{(k)}F_{\Omega}(I,\hat{A}^{(k)})^\top(\hat{C}^{(k)})^\top\|_F}{\|F_\Omega(E,A)^\top C^\top\|_F}&=8.6597\times 10^{-9}.\nonumber
\end{align}
The decay of the normalized residual \eqref{fl_res} for the frequency-limited controllability Gramian \(P_{\Omega}\) and the frequency-limited observability Gramian \(Q_{\Omega}\) is plotted in Figure \ref{fig4}.
\begin{figure}[!h]
  \centering
  \includegraphics[width=12cm]{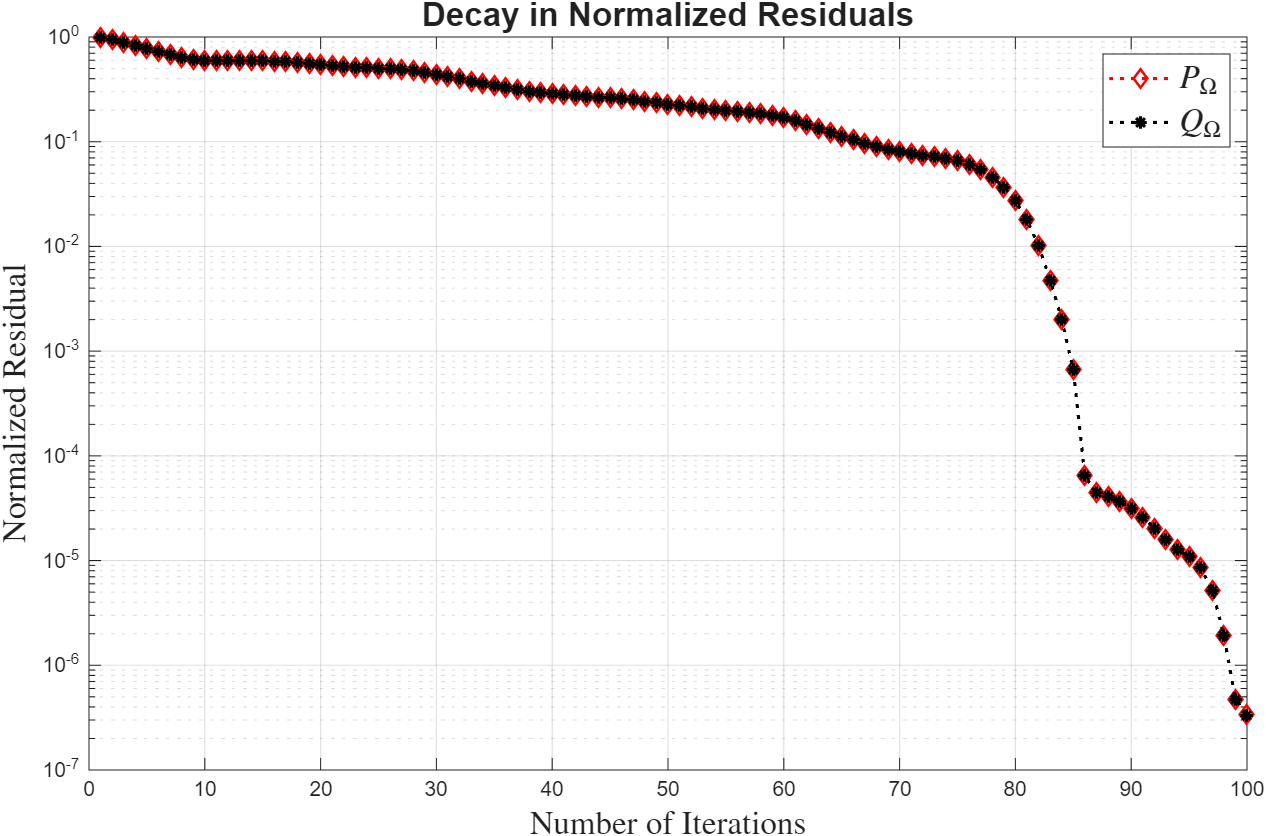}
  \caption{Decay of the normalized residual}\label{fig4}
\end{figure}

The minimum reduced order required to capture the three peaks is \(6\). In low-rank FLBT, the reduced order is set to \(r=12\) in this example. The singular values of the \(12^{\text{th}}\)-order ROM produced by low-rank FLBT are plotted in Figure \ref{fig5}. It can be seen that the ROM has peaks at \(10\), \(15\), and \(20\) rad/s, as intended.
\begin{figure}[!h]
  \centering
  \includegraphics[width=12cm]{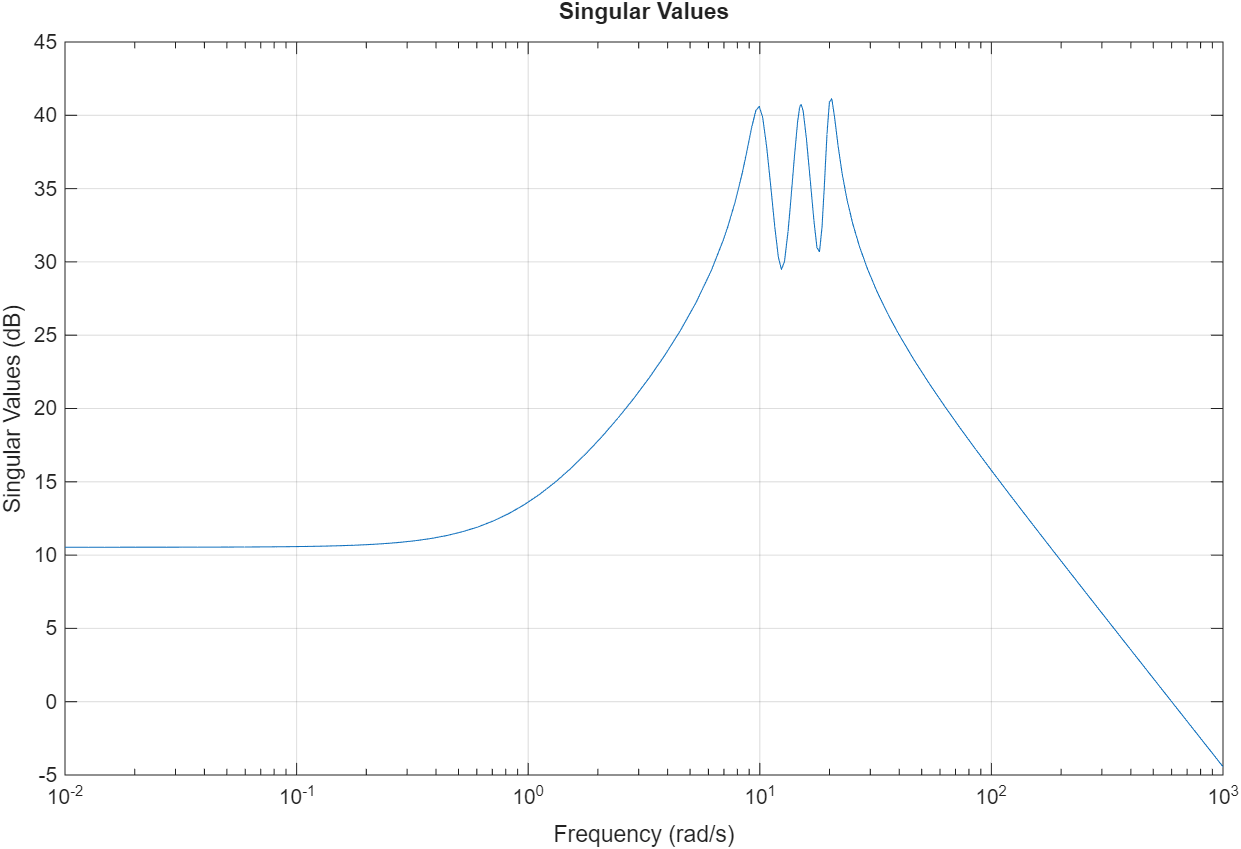}
  \caption{Singular values of the ROM}\label{fig5}
\end{figure}

The desired time interval for TLBT is set to \([0,2]\) s for demonstration purposes. The relative errors obtained with G-LRCF-ADI\(^{(\text{R})}\) for approximating \(e^{E^{-1}A\tau}E^{-1}B\) and \( e^{E^{-\top}A^\top \tau}E^{-\top}C^\top\) are given by
\begin{align}
\frac{\|e^{E^{-1}A\tau}E^{-1}B-V_{\mathrm{fadi}}^{(k)}e^{\hat{A}^{(k)}\tau}\hat{B}^{(k)}\|_F}{\|e^{E^{-1}A\tau}E^{-1}B\|_F}&=1.1423\times 10^{-5},\nonumber\\
\frac{\|e^{E^{-\top}A^\top \tau}E^{-\top}C^\top-Z_{\mathrm{fadi}}^{(k)}e^{(\hat{A}^{(k)})^\top \tau}(\hat{C}^{(k)})^\top\|_F}{\|e^{E^{-\top}A^\top \tau}E^{-\top}C^\top\|_F}&=1.1423\times 10^{-5}.\nonumber
\end{align}
The decay of the normalized residual \eqref{tl_res} for the time-limited controllability Gramian \(P_{\tau}\) and the time-limited observability Gramian \(Q_{\tau}\) is plotted in Figure \ref{fig6}.
\begin{figure}[!h]
  \centering
  \includegraphics[width=12cm]{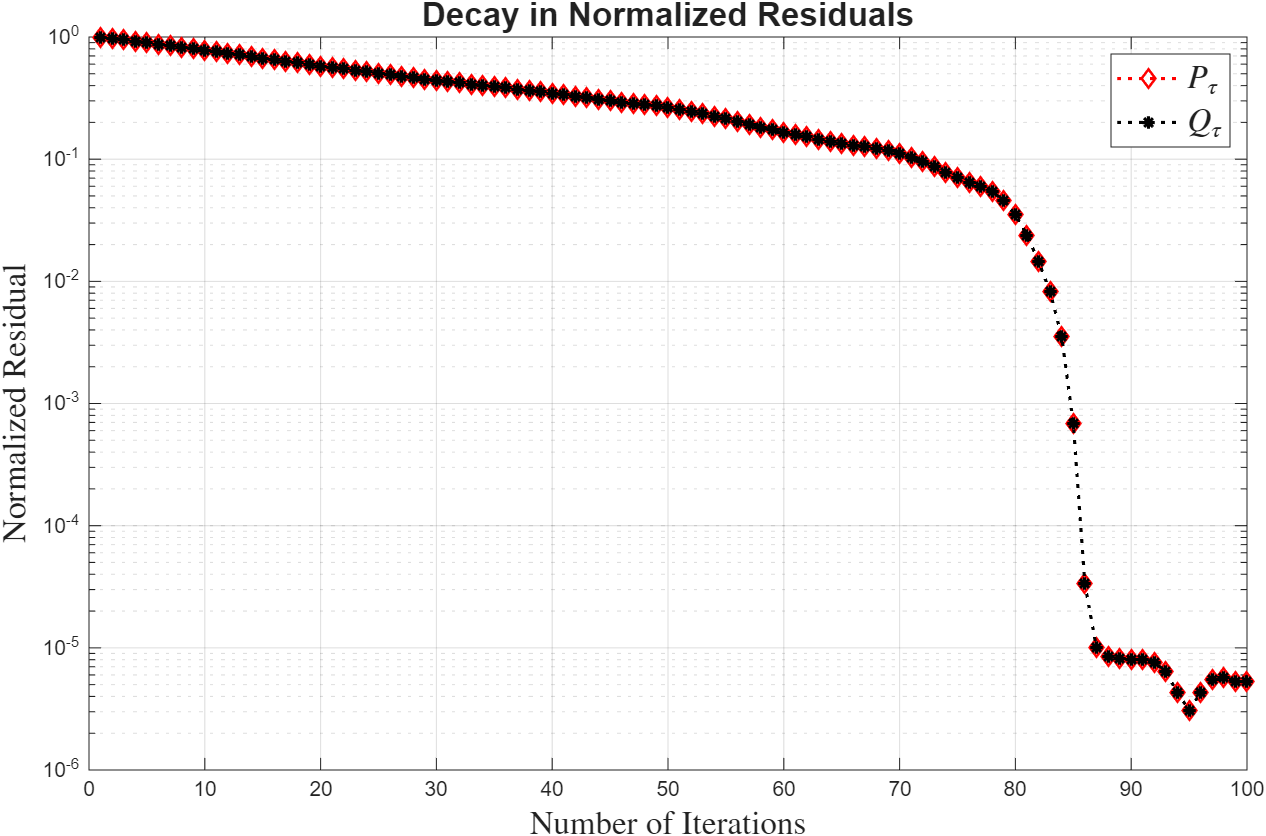}
  \caption{Decay of the normalized residual}\label{fig6}
\end{figure}

The singular values of the \(12^{\text{th}}\)-order ROM produced by low-rank TLBT are plotted in Figure \ref{fig7}. It can be seen that the ROM also has peaks at \(10\), \(15\), and \(20\) rad/s.
\begin{figure}[!h]
  \centering
  \includegraphics[width=12cm]{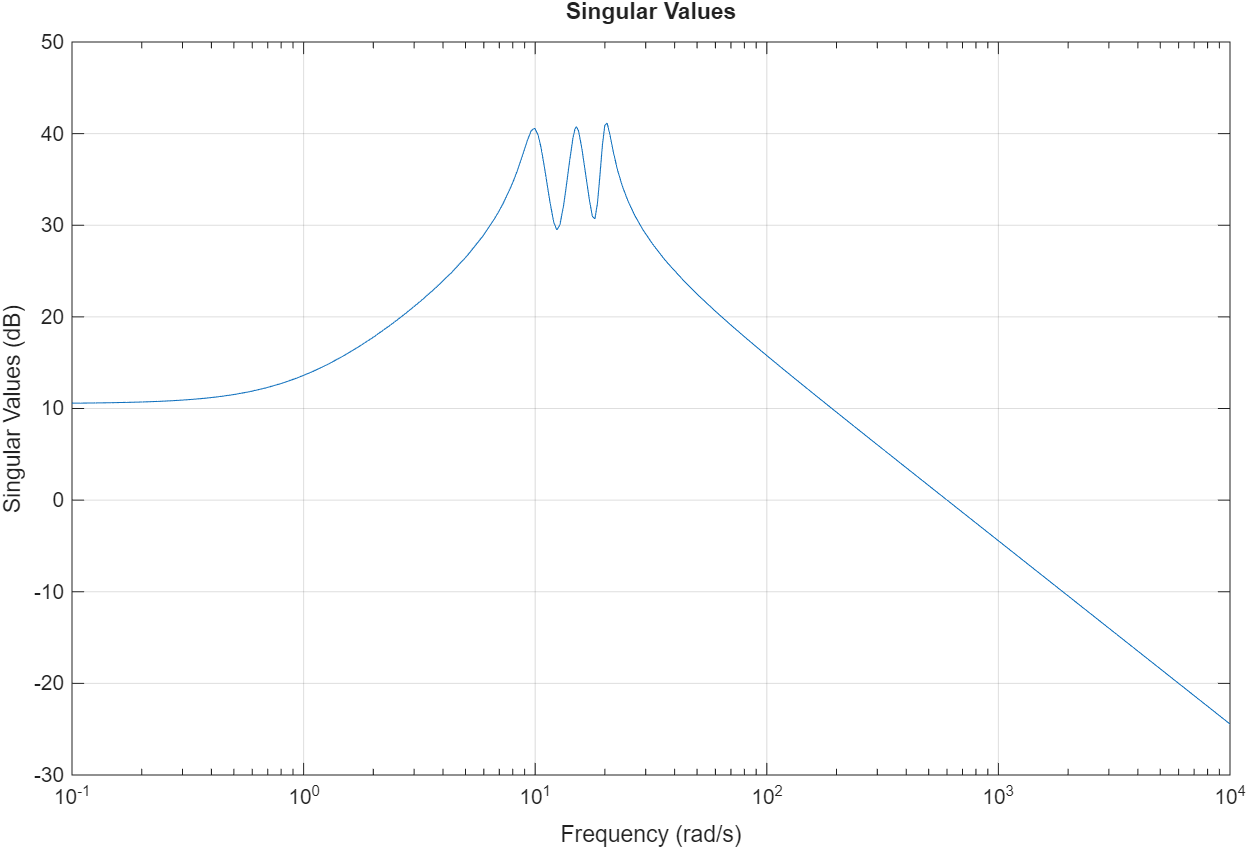}
  \caption{Singular values of the ROM}\label{fig7}
\end{figure}
\subsection{Data-driven reduced-order modeling}
Set \(n=500\), \(w_1=20\), \(w_2=50\), and \(w_3=90\) in \eqref{triple_peak} to generate the dynamical system used to test the data-driven ADI framework. This model has six dominant Hankel singular values associated with the six poles \(-1\pm j20\), \(-1\pm j50\), and \(-1\pm j90\). The remaining Hankel singular values are relatively insignificant. The 50 ADI shifts \(\alpha_i\) and the 50 ADI shifts \(\mu_i\) are purely imaginary, with imaginary parts chosen as exponentially spaced frequency points in \([10,100]\) rad/s so that this interval includes the three peaks at \(20\), \(50\), and \(90\) rad/s. The prescribed poles \(\beta_i\) and \(\nu_i\) are set to \(\beta_i=-1+\alpha_i\) and \(\nu_i=-1+\mu_i\). The \(20^{\text{th}}\)-order ROMs are constructed using BT and ADI-based data-driven BT. The Hankel singular values of the ROMs are compared in Figure \ref{fig8}. It can be seen that the data-driven BT has successfully captured the six dominant Hankel singular values of the original model.
\begin{figure}[!h]
  \centering
  \includegraphics[width=12cm]{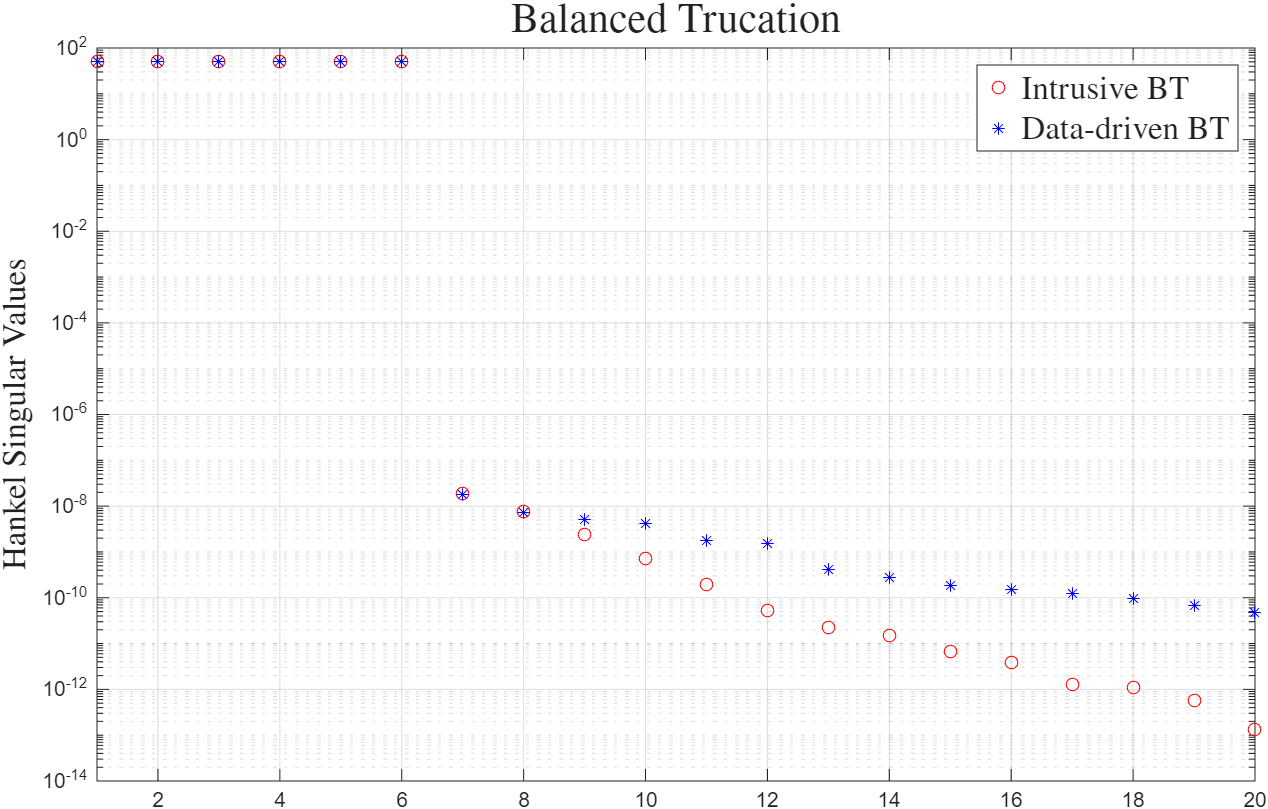}
  \caption{Comparison of Hankel singular values}\label{fig8}
\end{figure}

Next, ROMs \(G_r(s)\) of orders \(1\) to \(20\) are constructed using BT and ADI-based data-driven BT. The relative error \(\frac{\|G(s)-G_r(s)\|_{\mathcal{H}_\infty}}{\|G(s)\|_{\mathcal{H}_\infty}}\) is plotted in Figure \ref{fig9}. It can be seen that once the order reaches six, which is the minimum order required to capture the six dominant Hankel singular values, the relative error drops significantly, demonstrating that the data-driven BT has captured those dominant Hankel singular values. 
\begin{figure}[!h]
  \centering
  \includegraphics[width=12cm]{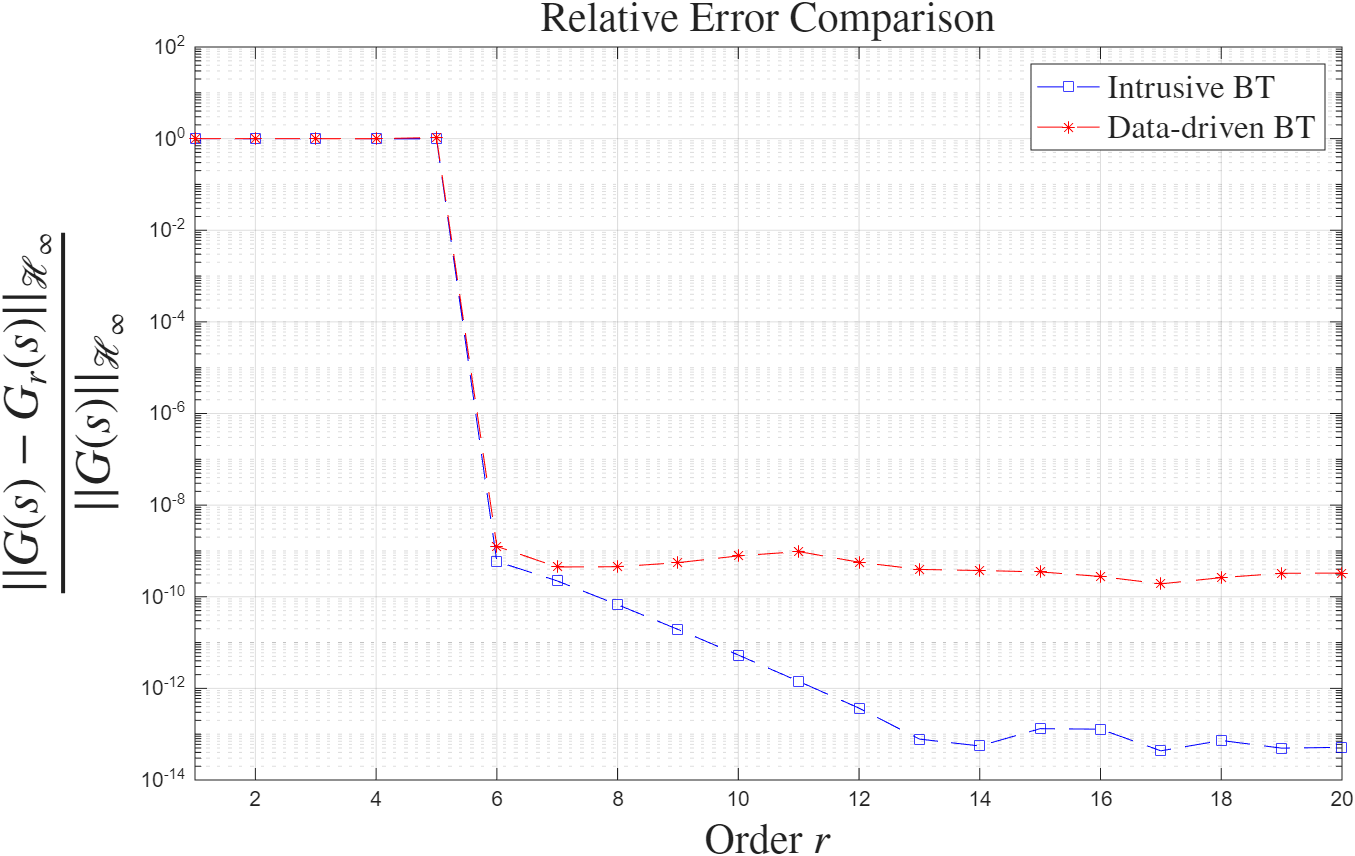}
  \caption{Comparison of relative errors}\label{fig9}
\end{figure}
\section{Conclusion}
A new Cholesky-factor ADI algorithm is proposed for computing low-rank solutions of large-scale Lyapunov equations. Unlike the existing algorithm, which requires the ADI shifts to lie in the left half of the complex plane, the proposed algorithm allows the shifts to be located anywhere in the complex plane. Owing to this property, the proposed algorithm can also be used as a numerical integration method to compute low-rank solutions of frequency-limited Lyapunov equations. Its application is further extended to time-limited Lyapunov equations and Riccati equations. The connection between Krylov-subspace-based rational interpolation and the proposed algorithm is also established. Moreover, the proposed algorithm can be used for data-driven reduced-order modeling, producing ROMs with several interesting system properties, such as pole and zero preservation. The numerical performance of the proposed algorithm is evaluated using benchmark numerical examples, and the results confirm its efficacy.
\section*{Appendix A: Derivation of the Residual Expression for Riccati Equations}
\begin{proof}
Recall that
\begin{align}
BB^\top=B_{\perp,k}^{\mathrm{fadi}}(B_{\perp,k}^{\mathrm{fadi}})^\top+B(\hat{B}^{(k)})^\top (V_{\mathrm{fadi}}^{(k)})^\top E^\top
+EV_{\mathrm{fadi}}^{(k)}\hat{B}^{(k)}B^\top-EV_{\mathrm{fadi}}^{(k)}\hat{B}^{(k)} (\hat{B}^{(k)})^\top (V_{\mathrm{fadi}}^{(k)})^\top E^\top.\nonumber
\end{align}
Next, consider the following
\begin{align}
&AV_{\mathrm{fadi}}^{(k)}\hat{P}_{\mathrm{ricc}}^{(k)}(V_{\mathrm{fadi}}^{(k)})^\top E^\top+EV_{\mathrm{fadi}}^{(k)}\hat{P}_{\mathrm{ricc}}^{(k)}(V_{\mathrm{fadi}}^{(k)})^\top A^\top +BB^\top -EV_{\mathrm{fadi}}^{(k)}\hat{P}_{\mathrm{ricc}}^{(k)}(V_{\mathrm{fadi}}^{(k)})^\top C^\top CV_{\mathrm{fadi}}^{(k)}\hat{P}_{\mathrm{ricc}}^{(k)}(V_{\mathrm{fadi}}^{(k)})^\top E^\top\nonumber\\
&=AV_{\mathrm{fadi}}^{(k)}\hat{P}_{\mathrm{ricc}}^{(k)}(V_{\mathrm{fadi}}^{(k)})^\top E^\top+EV_{\mathrm{fadi}}^{(k)}\hat{P}_{\mathrm{ricc}}^{(k)}(V_{\mathrm{fadi}}^{(k)})^\top A^\top +B_{\perp,k}^{\mathrm{fadi}}(B_{\perp,k}^{\mathrm{fadi}})^\top+B(\hat{B}^{(k)})^\top (V_{\mathrm{fadi}}^{(k)})^\top E^\top
+EV_{\mathrm{fadi}}^{(k)}\hat{B}^{(k)}B^\top\nonumber\\
&-EV_{\mathrm{fadi}}^{(k)}\hat{B}^{(k)} (\hat{B}^{(k)})^\top (V_{\mathrm{fadi}}^{(k)})^\top E^\top -EV_{\mathrm{fadi}}^{(k)}\hat{P}_{\mathrm{ricc}}^{(k)}(V_{\mathrm{fadi}}^{(k)})^\top C^\top CV_{\mathrm{fadi}}^{(k)}\hat{P}_{\mathrm{ricc}}^{(k)}(V_{\mathrm{fadi}}^{(k)})^\top E^\top\nonumber
\end{align}
Furthermore, recall that
\[
\hat{P}_{\mathrm{ricc}}^{(k)}(V_{\mathrm{fadi}}^{(k)})^\top C^\top CV_{\mathrm{fadi}}^{(k)}\hat{P}_{\mathrm{ricc}}^{(k)}=\hat{A}^{(k)}\hat{P}_{\mathrm{ricc}}^{(k)}+\hat{P}_{\mathrm{ricc}}^{(k)}(\hat{A}^{(k)})^\top +\hat{B}^{(k)}(\hat{B}^{(k)})^\top.
\]
Thus
\begin{align}
&AV_{\mathrm{fadi}}^{(k)}\hat{P}_{\mathrm{ricc}}^{(k)}(V_{\mathrm{fadi}}^{(k)})^\top E^\top+EV_{\mathrm{fadi}}^{(k)}\hat{P}_{\mathrm{ricc}}^{(k)}(V_{\mathrm{fadi}}^{(k)})^\top A^\top +BB^\top -EV_{\mathrm{fadi}}^{(k)}\hat{P}_{\mathrm{ricc}}^{(k)}(V_{\mathrm{fadi}}^{(k)})^\top C^\top CV_{\mathrm{fadi}}^{(k)}\hat{P}_{\mathrm{ricc}}^{(k)}(V_{\mathrm{fadi}}^{(k)})^\top E^\top\nonumber\\
&=\Big(AV_{\mathrm{fadi}}^{(k)}-EV_{\mathrm{fadi}}^{(k)}\hat{A}^{(k)}\Big)\hat{P}_{\mathrm{ricc}}^{(k)}(V_{\mathrm{fadi}}^{(k)})^\top E^\top+EV_{\mathrm{fadi}}^{(k)}\hat{P}_{\mathrm{ricc}}^{(k)}\Big((V_{\mathrm{fadi}}^{(k)})^\top A^\top-(\hat{A}^{(k)})^\top (V_{\mathrm{fadi}}^{(k)})^\top E^\top\Big)\nonumber\\
&+B_{\perp,k}^{\mathrm{fadi}}(B_{\perp,k}^{\mathrm{fadi}})^\top+B(\hat{B}^{(k)})^\top (V_{\mathrm{fadi}}^{(k)})^\top E^\top+EV_{\mathrm{fadi}}^{(k)}\hat{B}^{(k)}B^\top-2EV_{\mathrm{fadi}}^{(k)}\hat{B}^{(k)} (\hat{B}^{(k)})^\top (V_{\mathrm{fadi}}^{(k)})^\top E^\top\nonumber\\
&=-B_{\perp,k}^{\mathrm{fadi}}L_{\mathrm{fadi}}^{(k)}\hat{P}_{\mathrm{ricc}}^{(k)}(V_{\mathrm{fadi}}^{(k)})^\top E^\top-EV_{\mathrm{fadi}}^{(k)}\hat{P}_{\mathrm{ricc}}^{(k)}(L_{\mathrm{fadi}}^{(k)})^\top (B_{\perp,k}^{\mathrm{fadi}})^\top+B_{\perp,k}^{\mathrm{fadi}}(B_{\perp,k}^{\mathrm{fadi}})^\top+B_{\perp,k}^{\mathrm{fadi}}(\hat{B}^{(k)})^\top (V_{\mathrm{fadi}}^{(k)})^\top E^\top\nonumber\\
&+EV_{\mathrm{fadi}}^{(k)}\hat{B}^{(k)}(B_{\perp,k}^{\mathrm{fadi}})^\top\nonumber\\
&=B_{\perp,k}^{\mathrm{fadi}}(B_{\perp,k}^{\mathrm{fadi}})^\top+B_{\perp,k}^{\mathrm{fadi}}(F_{\mathrm{ricc}}^{(k)})^\top+F_{\mathrm{ricc}}^{(k)}(B_{\perp,k}^{\mathrm{fadi}})^\top.\nonumber
\end{align}
\end{proof}
\section*{Appendix B: Connection Between fADI and Krylov Subspace}
\begin{proof}
For simplicity of presentation, we restrict the proof for $m=1$. Nevertheless, the proof for $m>1$ follows directly. Let $R_i = (-\alpha_i E - A)^{-1}$ denote the resolvent operator at shift $\alpha_i$. The standard rational Krylov basis is given by $V_{\mathrm{data}}^{(k)} = \begin{bmatrix} R_1 B & \cdots & R_k B \end{bmatrix}$. We seek an upper triangular matrix $T_v^{(k)}$ such that the $i$-th column of $V_{\mathrm{fadi}}^{(k)}$ satisfies:
\begin{equation}
    v_i^{\mathrm{fadi}} = \sum_{j=1}^{i} T_v^{(k)}(j,i) R_j B.
\end{equation}
Recall that
\[
B_{\perp,i}^{\mathrm{fadi}} = B - \sum_{q=1}^{i} (\alpha_q + \beta_q) E v_q^{\mathrm{fadi}}.
\]
Substituting $v_q^{\mathrm{fadi}} = \sum_{j=1}^{q} T_v^{(k)}(j,q) R_j B$ and interchanging the order of summation yields:
\begin{equation}
    B_{\perp,i}^{\mathrm{fadi}} = B - \sum_{j=1}^{i} \left( \sum_{q=j}^{i} (\alpha_q + \beta_q) T_v^{(k)}(j,q) \right) E R_j B.
\end{equation}
The next vector in the basis is generated via $v_{i+1}^{\mathrm{fadi}} = (A + \alpha_{i+1} E)^{-1} B_{\perp,i}^{\mathrm{fadi}} = -R_{i+1} B_{\perp,i}^{\mathrm{fadi}}$. Substituting the expansion of $B_{\perp,i}^{\mathrm{fadi}}$ gives:
\begin{equation}
    v_{i+1}^{\mathrm{fadi}} = -R_{i+1} B + \sum_{j=1}^{i} \left( \sum_{q=j}^{i} (\alpha_q + \beta_q) T_v^{(k)}(j,q) \right) R_{i+1} E R_j B.
\end{equation}
We now apply the resolvent identity for the pencil $(A,E)$, which states that $R_{i+1} - R_j = (\alpha_{i+1} - \alpha_j) R_{i+1} E R_j$, or equivalently:
\begin{equation}
    R_{i+1} E R_j = \frac{R_{i+1} - R_j}{\alpha_{i+1} - \alpha_j}.
\end{equation}
Substituting this identity into the expression for $v_{i+1}^{\mathrm{fadi}}$ and regrouping the terms by the basis vectors $R_j B$, we obtain:
\begin{equation}
    v_{i+1}^{\mathrm{fadi}} = \left( -1 + \sum_{j=1}^{i} \frac{\sum_{q=j}^{i} (\alpha_q + \beta_q) T_v^{(k)}(j,q)}{\alpha_{i+1} - \alpha_j} \right) R_{i+1} B + \sum_{j=1}^{i} \frac{\sum_{q=j}^{i} (\alpha_q + \beta_q) T_v^{(k)}(j,q)}{\alpha_j - \alpha_{i+1}} R_j B.
\end{equation}
By comparing coefficients with $v_{i+1}^{\mathrm{fadi}} = \sum_{j=1}^{i+1} T_v^{(k)}(j, i+1) R_j B$, we extract the recurrence relation for the off-diagonal entries ($j \le i$):
\begin{equation}
\label{eq:T_recurrence}
    T_v^{(k)}(j, i+1) = \frac{1}{\alpha_j - \alpha_{i+1}} \sum_{q=j}^{i} (\alpha_q + \beta_q) T_v^{(k)}(j,q).
\end{equation}
We verify that the closed-form expression \eqref{eq:T_closed_form} satisfies the recurrence \eqref{eq:T_recurrence} and the initial condition. 
For the base case $i=1$, the algorithm initializes $B_{\perp,0}^{\mathrm{fadi}} = B$, yielding $v_1^{\mathrm{fadi}} = -R_1 B$. Thus, $T_v^{(k)}(1,1) = -1$, which matches \eqref{eq:T_closed_form} for $i=1, j=1$.

For the inductive step, assume \eqref{eq:T_closed_form} holds for all columns up to $i$. Substituting the inductive hypothesis into the right-hand side of \eqref{eq:T_recurrence} yields:
\begin{equation}
    T_v^{(k)}(j, i+1) = \frac{1}{\alpha_j - \alpha_{i+1}} \sum_{q=j}^{i} (\alpha_q + \beta_q) (-1)^{q} \frac{\prod_{l=1}^{q-1}(\alpha_j + \beta_l)}{\prod_{\substack{l=1 \\ l \neq j}}^{q}(\alpha_l - \alpha_j)}.
\end{equation}
Through algebraic manipulation and the properties of partial fraction decompositions (or equivalently, by recognizing the sum as a divided difference of the polynomial $p(x) = \prod_{l=1}^{i}(x + \beta_l)$), the summation simplifies exactly to:
\begin{equation}
    \sum_{q=j}^{i} (\alpha_q + \beta_q) (-1)^{q} \frac{\prod_{l=1}^{q-1}(\alpha_j + \beta_l)}{\prod_{\substack{l=1 \\ l \neq j}}^{q}(\alpha_l - \alpha_j)} = (-1)^{i+1} \frac{\prod_{l=1}^{i}(\alpha_j + \beta_l)}{\prod_{\substack{l=1 \\ l \neq j}}^{i}(\alpha_l - \alpha_j)}.
\end{equation}
Dividing by $(\alpha_j - \alpha_{i+1})$ yields:
\begin{equation}
    T_v^{(k)}(j, i+1) = (-1)^{i+1} \frac{\prod_{l=1}^{i}(\alpha_j + \beta_l)}{\prod_{\substack{l=1 \\ l \neq j}}^{i+1}(\alpha_l - \alpha_j)},
\end{equation}
which is precisely the closed-form expression \eqref{eq:T_closed_form} evaluated at column $i+1$. This completes the induction.
\end{proof}

\begin{thebibliography}{10}
\expandafter\ifx\csname url\endcsname\relax
  \def\url#1{\texttt{#1}}\fi
\expandafter\ifx\csname urlprefix\endcsname\relax\def\urlprefix{URL }\fi
\expandafter\ifx\csname href\endcsname\relax
  \def\href#1#2{#2} \def\path#1{#1}\fi

\bibitem{benner2013efficient}
P.~Benner, P.~K{\"u}rschner, J.~Saak, Efficient handling of complex shift
  parameters in the low-rank Cholesky factor ADI method, Numerical Algorithms
  62~(2) (2013) 225--251.

\bibitem{benner2013reformulated}
P.~Benner, P.~K{\"u}rschner, J.~Saak, A reformulated low-rank ADI iteration
  with explicit residual factors, PAMM 13~(1) (2013) 585--586.

\bibitem{tu2009adi}
P.~Benner, R.-C. Li, N.~Truhar, On the ADI method for Sylvester equations,
  Journal of Computational and Applied Mathematics 233~(4) (2009) 1035--1045.

\bibitem{benner2014computing}
P.~Benner, P.~K{\"u}rschner, Computing real low-rank solutions of Sylvester
  equations by the factored ADI method, Computers \& Mathematics with
  Applications 67~(9) (2014) 1656--1672.

\bibitem{zulfiqar2026unified}
U.~Zulfiqar, Z.-Y. Huang, Q.-Y. Song, Z.-Y. Gao, A unified low-rank ADI
  framework with shared linear solves for simultaneously solving multiple
  Lyapunov, Sylvester, and Riccati equations, Journal of Computational and
  Applied Mathematics (2026) 117897.

\bibitem{mycodes}
U.~Zulfiqar, \href{https://doi.org/10.5281/zenodo.21522095}{{MATLAB} codes for
  ``A new low-rank Cholesky-factor ADI algorithm allowing shifts anywhere in
  the complex plane with applications to data-driven model reduction''},
  https://doi.org/10.5281/zenodo.21522095 (2026).
\newline\urlprefix\url{https://doi.org/10.5281/zenodo.21522095}

\bibitem{wolf2014h}
T.~Wolf, $\mathcal{H}_2$ pseudo-optimal model order reduction, Ph.D. thesis, Technische
  Universit{\"a}t M{\"u}nchen (2014).

\bibitem{gawronski1990model}
W.~Gawronski, J.-N. Juang, Model reduction in limited time and frequency
  intervals, International Journal of Systems Science 21~(2) (1990) 349--376.

\bibitem{petersson2014model}
D.~Petersson, J.~L{\"o}fberg, Model reduction using a frequency-limited
  $\mathcal{H}_2$-cost, Systems \& Control Letters 67 (2014) 32--39.

\bibitem{zulfiqar2022adaptive}
U.~Zulfiqar, V.~Sreeram, X.~Du, Adaptive frequency-limited-model order
  reduction, Asian Journal of Control 24~(6) (2022) 2807--2823.

\bibitem{tombs1987truncated}
M.~S. Tombs, I.~Postlethwaite, Truncated balanced realization of a stable
  non-minimal state-space system, International Journal of Control 46~(4)
  (1987) 1319--1330.

\bibitem{benner2018radi}
P.~Benner, Z.~Bujanovi{\'c}, P.~K{\"u}rschner, J.~Saak, RADI: A low-rank
  ADI-type algorithm for large scale algebraic Riccati equations, Numerische
  Mathematik 138~(2) (2018) 301--330.

\bibitem{zulfiqar2026data}
U.~Zulfiqar, Q.-Y. Song, Z.-H. Xiao, V.~Sreeram, Data-driven implementations of
  various generalizations of balanced truncation, Computational and Applied
  Mathematics 45~(9) (2026) 394.

\bibitem{zulfiqar2026proj}
U.~Zulfiqar, From data $H(j\omega_i)$ to balanced truncation family: A
  projection-based non-intrusive approach, arXiv preprint arXiv:2602.12697
  (2026).

\bibitem{mayo2007framework}
A.~Mayo, A.~C. Antoulas, A framework for the solution of the generalized
  realization problem, Linear Algebra and Its Applications 425~(2-3) (2007)
  634--662.

\bibitem{jonckheere1983new}
E.~Jonckheere, L.~Silverman, A new set of invariants for linear
  systems--Application to reduced order compensator design, IEEE Transactions
  on Automatic Control 28~(10) (1983) 953--964.

\bibitem{zhou1995frequency}
K.~Zhou, Frequency-weighted $\mathcal{L}_\infty$ norm and optimal Hankel norm
  model reduction, IEEE Transactions on Automatic Control 40~(10) (1995)
  1687--1699.

\bibitem{benner2005semi}
P.~Benner, J.~Saak, A semi-discretized heat transfer model for optimal cooling
  of steel profiles, in: Dimension Reduction of Large-Scale Systems:
  Proceedings of a Workshop held in Oberwolfach, Germany, October 19--25, 2003,
  Springer, 2005, pp. 353--356.

\bibitem{morPen06}
T.~Penzl, Algorithms for model reduction of large dynamical systems, Linear
  Algebra and Its Application 415~(2--3) (2006) 322--343.
\newblock \href {https://doi.org/10.1016/j.laa.2006.01.007}
  {\path{doi:10.1016/j.laa.2006.01.007}}.

\bibitem{chahlaoui2005benchmark}
Y.~Chahlaoui, P.~Van~Dooren, Benchmark examples for model reduction of linear
  time-invariant dynamical systems, in: Dimension Reduction of Large-Scale
  Systems: Proceedings of a Workshop held in Oberwolfach, Germany, October
  19--25, 2003, Springer, 2005, pp. 379--392.

\end{thebibliography}

\end{document}